\documentclass[onefignum,onetabnum]{siamart190516}

\usepackage[UKenglish]{babel}
\usepackage{lipsum}
\usepackage{amsfonts}
\usepackage{amssymb}
\usepackage{amsmath}
\usepackage{graphicx}
\usepackage{epstopdf}
\usepackage{algorithmic}
\usepackage[backend=bibtex, style=phys, sorting=nyt, biblabel=brackets]{biblatex}
\usepackage[backend=bibtex]{biblatex}
\usepackage{todonotes}
\usepackage{enumerate}
\usepackage{chngpage}
\usepackage{diagbox}
\usepackage{amsopn}
\usepackage[figuresright]{rotating}
\usepackage{makecell}
\usepackage{xcolor}
\usepackage{caption} 

\usepackage{cellspace} 
\definecolor{commentcolor}{rgb}{0,0,0}

\ifpdf
\DeclareGraphicsExtensions{.eps,.pdf,.png,.jpg}
\else
\DeclareGraphicsExtensions{.eps}
\fi

\newsiamremark{remark}{Remark}
\newsiamremark{hypothesis}{Hypothesis}
\crefname{hypothesis}{Hypothesis}{Hypotheses}
\newsiamthm{claim}{Claim}

\headers{Sparse Tensor Product Approximation for GMM estimators }{A. Gilch, M. Griebel, J. Oettershagen}

\title{Sparse Tensor Product Approximation for a class of Generalized Method of Moments Estimators}

\author{Alexandros Gilch\thanks{Institute for Numerical Simulation, Bonn, Germany, 	\url{https://ins.uni-bonn.de/group/griebel/profile}
		(\email{gilch@ins.uni-bonn.de}, \email{griebel@ins.uni-bonn.de}, \email{oettershagen@ins.uni-bonn.de}).}
	\and Michael Griebel\footnotemark[1] \thanks{Fraunhofer Institute for Algorithms and Scientific Computing SCAI, Sankt Augustin, Germany, \url{https://www.scai.fraunhofer.de/en.html}}
	\and Jens Oettershagen\footnotemark[1]}

\renewcommand{\th}{\theta}

\DeclareMathOperator{\R}{\mathbb{R}}

\DeclareMathOperator{\N}{\mathbb{N}}

\DeclareMathOperator*{\argmax}{\operatorname{argmax}}

\DeclareMathOperator{\ra}{\operatorname{\rightarrow}}
\DeclareMathOperator{\E}{\operatorname{E}}
\DeclareMathOperator{\Normal}{\operatorname{\mathcal{N}}}

\newcommand{\OG}{G_\infty}

\newcommand{\OQ}{R_\infty}

\newcommand{\SobolevspaceI}{H_{\text{mix}}^{r_i}}

\DeclareMathOperator{\Productspace}{\Omega_1\times\Omega_2}
\DeclareMathOperator{\Bochnerspace}{\mathcal{H}_1\left(\Omega_1,\mathcal{H}_2(\Omega_2;\mu);\nu\right)}
\DeclareMathOperator{\OM}{\Omega_1}
\DeclareMathOperator{\OMM}{\Omega_2}

\DeclareMathOperator{\Int}{\mathcal{I}}

\newcommand{\QFG}{Q_{\infty,L}^\sigma}
\newcommand{\QSG}{Q_{1,L}^\sigma}
\newcommand{\NFG}{N_{\infty,L}^\sigma}
\newcommand{\NSG}{N_{1,L}^\sigma}
\newcommand{\ESG}{E_{1,L}^\sigma}
\newcommand{\EFG}{E_{\infty,L}^\sigma}
\DeclareMathOperator{\innerintegrand}{\varphi}

\DeclareMathOperator{\polylog}{\operatorname{Li}}
\newcommand{\Quad}{Q^i}
\newcommand{\QuadI}{Q^1}
\newcommand{\QuadII}{Q^2}
\newcommand{\DQuad}{\Delta^i}
\newcommand{\DQuadI}{\Delta^1}
\newcommand{\DQuadII}{\Delta^2}
\DeclareMathOperator{\LSFG}{\operatorname{\mathcal{A}^\sigma_\infty}}
\DeclareMathOperator{\LSSG}{\operatorname{\mathcal{A}^\sigma_1}}
\newcommand{\Qi}{Q^i_{l_i}}
\newcommand{\Ni}{N_{i,l_i}}
\newcommand{\gFGopt}{\gamma_\infty^*}
\newcommand{\gSGopt}{\gamma_1^*}
\newcommand{\wni}{w_{n_i}^{i,l_i}}

\DeclareMathOperator{\arcsinh}{arcsinh}
\DeclareMathOperator{\arctanh}{arctanh}
\DeclareMathOperator{\erf}{erf}

\ifpdf
\hypersetup{
  pdftitle={Sparse Tensor Product Approximation for a class of Generalized Method of Moments Estimators},
  pdfauthor={A. Gilch, M. Griebel, J. Oettershagen}
}
\fi
\begin{document}

	\maketitle
	
	% REQUIRED
	\begin{abstract}
		Generalized Method of Moments (GMM) estimators in their various forms, including the popular Maximum  Likelihood (ML) estimator, are frequently applied for the evaluation of complex econometric models with not analytically computable moment or likelihood functions. As the objective functions of GMM- and ML-estimators themselves constitute the approximation of an integral, more precisely of the expected value over the real world data space, the question arises whether the approximation of the moment function and the simulation of the entire objective function can be combined. 
		\\
		Motivated by the popular Probit and Mixed Logit models, we consider double integrals with a linking function which stems from the considered estimator, e.g. the logarithm for Maximum Likelihood, and apply a sparse tensor product quadrature to reduce the computational effort for the approximation of the combined integral. Given H\"older continuity of the linking function, we prove that this approach can improve the order of the convergence rate of the classical GMM- and ML-estimator by a factor of two, even for integrands of low regularity or high dimensionality. This result is illustrated by numerical simulations of Mixed Logit and Multinomial Probit integrals which are estimated by ML- and GMM-estimators, respectively.
	\end{abstract}
	
	% REQUIRED
	
	% REQUIRED
	
	\section{Introduction}
		In the past decades, econometric models and the corresponding parametrizations and estimators have become increasingly challenging in terms of mathematical and computational complexity. In this context, Generalized Methods of Moments estimators and especially the Maximum Likelihood estimator serve as reliable tools to validate theoretically developed models against real world data. Both estimators can be written via an optimization principle as extremum estimators and are therefore determined by an objective function $R_N$ which needs to be maximized. For parameterized models with finite-dimensional parameter vector $\th$, the extremum estimator $\hat{\th}$ is defined as
		\begin{align} \label{ExtremumEstimator}
			\hat{\th}:=\argmax_{\th\in\Theta}\ R_N(\th)\,,
		\end{align}
		where $\Theta\subset\R^q$ denotes the examined parameter space.
		\\
		The exact form of $R_N$ is subject to the model formulation and incorporates the observed data points $z_1,...,z_N\in\mathcal{Z}$. The real world data space $\mathcal{Z}$ contains all possible data points in the considered econometric model which are distributed w.r.t.\ the probability distribution $\nu$ and is hence a subset of a finite-dimensional vector space. 
		\\
		\textcolor{commentcolor}{Given $\Theta$ and $\mathcal{Z}$, the GMM estimator is theoretically defined as the solution of the estimating equations
		\begin{align*}
			0=\E_z[m(z|\th)]=\int_{\mathcal{Z}}m(z|\th)d\nu(z)
		\end{align*}
		for some \textit{moment function} $m:\mathcal{Z}\times\Theta\ra\R^q$}. As $\nu$ is unknown, this integral can usually not be solved analytically so it is approximated by a sum over the observations
		\begin{align} \label{Objective GMM}
			G_N(\th)=\frac{1}{N}\sum_{n=1}^{N}m(z_n|\th)\,.
		\end{align}
		which gives rise to the GMM objective function in (\ref{ExtremumEstimator}) by
		\begin{align} \label{Objective ML}
			R^\text{GMM}_N(\th):=-||G_N(\th)||\,.
		\end{align}
		The observations are assumed to be random samples from a real world distribution and hence $G_N$ can be interpreted as Monte Carlo simulation of an expected value over the real world data space $\mathcal{Z}$,
		\begin{align} \label{IntegratedEstimator}
			G_N(\th)\approx G_\infty:=\int_{\mathcal{Z}}m(z|\th)d\nu(z)\,.
		\end{align}
		In particular, the approximation error is the same as for classical MC simulations, i.e. it is of order $O(N^{-1/2})$ in expectation, provided that the variance of $m$ is finite. Hence a high number of observations $N$ is required to obtain a well-behaved estimator.
		\\
		\textcolor{commentcolor}{The ML estimator is derived directly as maximizer in the sense of (\ref{ExtremumEstimator}), where 
		\begin{align*}
			R^\text{ML}_N(\th):=G_N(\th)
		\end{align*}
		is called the \textit{log-likelihood function} and $m(z_n|\th):=\log(f(z_n|\th))$ for a problem-specific function $f$. Here, $G_N$ is understood as the logarithm of the \textit{likelihood} given the data $z_1,...,z_N$ and the parameter $\th$, i.e.
		\begin{align*}
			\mathcal{L}(z_1,...,z_N)=\prod_{n=1}^{N}f(z_n|\th)\,,
		\end{align*}
		which one would like to maximize w.r.t. $\th$. Hence, this fits the setting of extremum estimation, providing us the same intuition of $G_N$ for GMM and ML estimators as sample analogs or MC simulations of the integral (\ref{IntegratedEstimator}). For a thorough introduction of GMM and ML estimators we refer to the textbooks by Hayashi \cite{hayashi2000econometrics} and Newey and McFadden \cite{NeweyMcFadden1994}.}
		\\
		While models are often designed to yield closed form objective functions, in many areas of economic research such formulations are not possible and the resulting expressions do not have a closed form. Important examples are stochastic dynamic models, where multidimensional partial differential equations need to be solved \cite{FERNANDEZVILLAVERDE2016}, \cite{KruegerKubler2004}, \cite{Maliar2014} or large state spaces have to be searched \cite{keane1994solution}, \cite{keane2011structural} and Generalized Linear Mixed models, which require the computation of multidimensional integrals \cite{GourierouxMonfort1997}, \cite{Train2009}.
		\\
		To this end, simulation or Monte Carlo (MC) techniques are the usual choice for such approximation problems since they are quite robust and fairly easy to implement.
		In particular they do not suffer from the curse of dimensionality: For tensor products of one-dimensional approximation rules, the computational costs of interpolation and integration grow exponentially in the dimensionality of the problem. Consequently, the convergence rate deteriorates exponentially with rising dimension $d$, leading to an infeasibility of models with many individuals, countries, choices, etc.. In contrast, MC sampling only provides a probabilistic convergence rate of order $O(N^{-1/2})$ for $N$ interpolation or quadrature nodes, which is however basically independent of $d$.
		\\
		In addition, numerical mathematics offers tools for handling multidimensional problems and circumventing the curse of dimensionality to some extent such as Quasi Monte Carlo (QMC) and Sparse Grid (SG, also known as Smolyak) methods. Other than MC, QMC and SG methods are based on deterministically computed interpolation or quadrature nodes and hence yield deterministic convergence rates. Given sufficient regularity of the approximated functions or integrands, i.e. involving bounded mixed derivatives, they achieve algebraic or even exponential convergence rates and thus can accelerate the evaluation of econometric models significantly. The curse of dimensionality now only appears in logarithmic terms of $N$ in the convergence rates or disappears even completely for certain anisotropic smoothness classes. For further details on SG methods see \cite{BungartzGriebel2004}, \cite{GerstnerGriebel1998}.
		\\
		Previous attempts to utilize the strength of modern numerical methods in econometrics support this prospect: Kr\"uger, K\"ubler and Malin \cite{KruegerKubler2004}, \cite{Malin2011} were the first to apply a Sparse Grid technique to find the dynamic equilibrium in an overlapping generations model by global polynomial interpolation. Similar approaches were adopted by Judd et al. \cite{Judd2014} and Winschel and Kraetzig \cite{WinschelKraetzig2010} for stochastic growth models. More recently, Brumm and Scheidegger \cite{BrummScheidegger2017} implemented an adaptive SG rule to compute global solutions of an International Real Business Cycle model.
		\\
		In another line of research, QMC and SG quadrature rules have been applied to approximate expected values and cumulative distribution functions. Such integrals often appear in likelihood functions when an unobservable variable is integrated out. Prominent examples are Probit and Mixed Logit models: While earlier works considered only one-dimensional rules \cite{butler1982computationally} and various simulation methods \cite{Geweke1998MC}, \cite{hajivassiliou1994classical}, Bhat \cite{Bhat2001} was the first to use QMC rules for numerical quadrature of a Mixed Logit model. Later, Heiss and Winschel \cite{Heiss2010}, \cite{HeissWinschel2008}, Judd and Skrainka \cite{JuddSkrainka2011} and Abay \cite{Abay2014} investigated the benefits of basic SG rules for the approximation of Panel Probit and multinomial Mixed Logit models. Griebel and Oettershagen \cite{GriebelOettershagen2014} developed an extension of SG quadrature which allows the numerical integration of boundary singularities and achieved exponential convergence for the Probit integral.
		\\
		In this paper, we consider objective functions for GMM- and ML-estimators where the function $m$ includes the approximation of an integral i.e.
		\begin{align} \label{MomentFunction}
			m(z,\th)=F\left(z,\th,\int_{\mathcal{U}}\varphi(z,u|\th)d\mu(u)\right)
		\end{align}
		for some functions $F:\mathcal{Z}\times\Theta\times\R\ra\R$ and $\varphi:\mathcal{Z}\times\mathcal{U}\times\Theta\ra\R$. Here, $\mathcal{U}$ is usually the domain of an unobserved error variable which is integrated out, i.e. $\mathcal{U}=\R^d$ for some $d\in\N$ with $\mu$ being the corresponding probability measure. 
		\\
		In this case, the ML-objective function can be simply derived by choosing $F(z,\th,t)=\log(t)$. \textcolor{commentcolor}{C}ombined with the integrated form of $R_N$ we obtain a double integral
		\begin{align} \label{DoubleIntegral}
			G_\infty(\th)=\int_{\mathcal{Z}}F\left(z,\th,\int_{\mathcal{U}}\varphi(z,u|\th)d\mu(u)\right)d\nu(z)
		\end{align}
		with linking function $F$. In this context, it becomes apparent that an improved accuracy of the approximation of the inner integral cannot compensate the sampling error that is inherited from the outer approximation. Instead, one needs to balance the inner approximation error with the outer sampling error in order to achieve the optimal error with the least possible computational effort. Griebel et al. \cite{GriebelOettershagen2019} recently gave an overview over the respective balancings for various quadrature rules including MC, QMC and SG rules.
		\\
		Moreover, the double integral (\ref{DoubleIntegral}) can be interpreted as a certain integral on the tensor product space $\mathcal{Z}\times \mathcal{U}$. Harbrecht and Griebel \cite{GriebelHarbrecht2011} investigated interpolation in such tensor product spaces and constructed a \textit{sparse tensor product space} based on the sparse grid method. Sparse tensor product spaces are used e.g. in \cite{GriebelHarbrechtMulterer2015} where elliptic PDEs are solved with quadrature on $\mathcal{Z}$ and interpolation on $\mathcal{U}$. Heinrich \cite{Heinrich2001} and Giles \cite{giles2015multilevel} proposed a similar approach called \textit{Multilevel Monte Carlo} method solely for Monte Carlo simulations of stochastic and deterministic PDEs and integrals, which are not analytically solvable. As first observed in \cite{GerstnerHeinz2013}, this approach just resembles the sparse tensor product approximation on the product of the integral space and the parameter space.
		\\
		We present in this paper a sparse tensor product quadrature rule which combines the inner approximation and the sampling by observations of the outer integral. The joint consideration of the two integration domain can be viewed as special form of a generalized sparse grid quadrature rule. We show that H\"older continuity of the linking function $F$ suffices to extend classical convergence results from sparse grid theory to sparse tensor product quadrature. In particular, we prove that the order of the convergence rate of an optimally balanced full tensor product rule \cite{GriebelOettershagen2019} can be improved by a factor of up to two with a properly chosen sparse tensor product approach. This is a substantial increase which allows to efficiently treat much more complicated problems in practice than just by a (well-balanced) quadrature rule as in \cite{GriebelOettershagen2019}.
		\\
		The remainder of this paper is organized as follows: In section \ref{MLE-setup}, we establish an integral representation of the objective function of GMM estimators and introduce associated notational conventions. In section \ref{MLE-theory} and inspired by this formulation, we develop a sparse tensor product quadrature and present results on convergence rates for functions from mixed regularity Sobolev spaces. In section \ref{DCM}, we present three exemplary econometric models from Discrete Choice theory which indeed possess the double integral structure introduced in (\ref{DoubleIntegral}). In section \ref{MLE-numerics}, we underscore our findings with numerical results for the previously established Mixed Logit model evaluated by Maximum Likelihood and the multinomial Probit model evaluated by a GMM-estimator. Finally, we give some concluding remarks in section \ref{Conclusion}.
		
	\section{Setup} \label{MLE-setup}
		We now consider the general functional
		\begin{align}\label{IntGMM}
			\OG(\th)=\int_{\mathcal{Z}}F\left(z,\th,\int_{\mathcal{U}}\varphi(z,u|\th)d\mu(u)\right)d\nu(z)\,.
		\end{align}
		Then, the asymptotic GMM estimator is given by
			\begin{align*}
				\OQ(\th)=||\OG(\th)||^2\,,
			\end{align*}
		where the function $F:\mathcal{Z}\times\Theta\times\R\ra\R$ is defined by the chosen moments $m$ and the variable $u$ often represents unobservable variables or errors in measurement. GMM-estimators are fairly general and also include Maximum Likelihood estimation as special case. It can simply be obtained with the linking function $F(z,\th,t)=\log(t)$.
		\\
		In the following, we assume that 
		\begin{enumerate}[(I)]
			\item $\varphi$ is $\mu$-integrable in $u$ for all $\theta\in\Theta$ and for $\nu$-almost all $z$ and
			\item $F$ is $\nu$-integrable in $z$ for all $\theta\in\Theta$.
		\end{enumerate}
		From now on, we omit the dependence on $\th$ as the integral is taken into account separately for each $\th$ and write (\ref{IntGMM}) with $\OM:=\mathcal{Z}$ and $\OMM:=\mathcal{U}$ more generally as
		\begin{align}
			\Int_1 (F_{\innerintegrand})&:=\int_{\mathcal{\OM}}F_{\innerintegrand}(z)d\nu(z)\,, \label{Int1}
			\\ 
			\Int_2 (\innerintegrand,z)&:=\int_{\mathcal{\OMM}}\innerintegrand(z,u)d\mu(u) \label{Int2}
		\end{align}
		for domains $\Omega_i\subset\R^{d_i}$, a $\nu$-integrable function $F_{\innerintegrand}:\OM\ra\R$ and a function $\innerintegrand:\Productspace\ra\R$ , where $\innerintegrand(z,\cdot)$ is $\mu$-integrable for every $z\in\OM$. We write $F_{\innerintegrand}$ to indicate that we consider functions $F$ which always include the computation of the integral $\Int_2$ but might also depend on $z$ in a direct way. We express this dependence via
		\begin{align*}
			F_{\innerintegrand}(z)=F(z,\th,\Int_2(\innerintegrand,z))\,.
		\end{align*}
		In order to apply quadrature rules like SG or QMC, certain regularity conditions have to be imposed on the considered integrands, whereas square-integrability is sufficient for MC quadrature. Since we examine integration over $z$ and $u$ separately at first, we need to determine separate regularity conditions. \textcolor{commentcolor}{Assume that $\mathcal{H}^1$ and $\mathcal{H}^2$ are separable Hilbert spaces.} We let $F(\cdot,\th,\cdot)\in\mathcal{H}_1(\OM\times\R,\nu)$ for every $\th\in\Theta$ and also $\Int_2(\innerintegrand,\cdot)\in\mathcal{H}_1(\OM,\nu)$ since this term is required to assemble $F_{\innerintegrand}$. In terms of the inner integral, we let $\innerintegrand(z,\cdot)\in\mathcal{H}_2(\OMM,\nu)$ for every $z\in\OM$. The assumptions on $\innerintegrand$ can be summarized to
		\begin{align*}
			\innerintegrand\in\mathcal{H}_1(\OM,\nu)\otimes\mathcal{H}_2(\OMM,\mu)\,,
		\end{align*}
		i.e. the inner integrand needs to have sufficient regularity in both domains, while the outer integrand $F_{\innerintegrand}$ only requires regularity in $z$. This resembles a mixed regularity assumption for $\innerintegrand$. \textcolor{commentcolor}{Here, $\mathcal{H}_1(\OM,\nu)\otimes\mathcal{H}_2(\OMM,\mu)$ denotes the tensor product of the two Hilbert spaces involving the usual metric space completion.}
		\\
		Two reasonable choices for $\mathcal{H}_1$ and $\mathcal{H}_2$ are to let them be $L^2$, the space of square-integrable functions, or $\SobolevspaceI$, the Sobolev space of mixed regularity for (possibly) different $r_i$, $i=1,2$, as both those spaces conform with the previously mentioned quadrature rules. Of course, other choices like classical Sobolev spaces are also possible.
		\\
		Another approach of fitting our problem into proper function spaces includes the rewriting of $\innerintegrand$ as $\innerintegrand(z,u)=\innerintegrand(z)(u)$. Then, we can view (\ref{Int1}) and (\ref{Int2}) for fixed $F$ as single integration problem of an integrand in
		\begin{align*}
			\mathcal{H}^F:=\Bigg\{\varphi\in\Bochnerspace\text{ s.t. }F\left(z,\th,\int_{\OMM}\innerintegrand(z)(u)d\mu(u)\right)\in \mathcal{H}_1(\OM,\nu)\Bigg\}\,.
		\end{align*}
		For $F_{\innerintegrand}(z)=\tilde{F}(|\Int_2|)$, i.e. the dependence on $z$ is solely via the integral $\Int_2$, this definition resembles the so-called \textit{Orlicz-Bochner space}. 
		If $\tilde{F}$ would be a so-called \textit{Young function}, see \cite{rao1991theory}, i.e. $\tilde{F}:\R\ra[0,\infty)$ is convex and lower semi-continuous with $\lim_{s\ra\infty}\frac{\tilde{F}(s)}{s}=\infty$ and $\lim_{s\ra 0}\frac{\tilde{F}(s)}{s}=0$, then the Orlicz-Bochner space is defined as 
		\begin{align*}
			L^{\tilde{F}}(\Omega,X):=\left\{\innerintegrand:\Omega\ra X \text{ measurable and }\exists\alpha>0:\int_{\Omega}\tilde{F}\big(\alpha||\innerintegrand(z)||_X\big)d\mu(z)<\infty\right\}\,.
		\end{align*}
		Setting $\Omega=\OM$ and $X=\mathcal{H}_2(\OMM,\mu)$ we get $\mathcal{H}^{\tilde{F}}\subset L^{\tilde{F}}$. This notation can be further extended to include weakly differentiable functions $\innerintegrand$, see \cite{AdamsFournier2003}.
		\\
		But the assumption that $F$ is a Young function and hence convex is too strong for the general form (\ref{IntGMM}). Already the case $F(z,\th,t)=\log(z)$ does not satisfy these conditions rendering a further investigation of Orlicz-Bochner spaces inadequate as Maximum Likelihood estimation is the most relevant application of our theoretical results. \textcolor{commentcolor}{We therefore proceed by examining $F_{\innerintegrand}$ and $\innerintegrand\in\mathcal{H}_1(\OM,\nu)\otimes\mathcal{H}_2(\OMM,\mu)$ and $F(\cdot,\th,\cdot)\in\mathcal{H}_1(\OM\times\R,\nu)$ for every $\th\in\Theta$ as already discussed before}.

	\section{Sparse Tensor Product Quadrature} \label{MLE-theory}
		We now intend to approximate the integrals $\Int_1$ and $\Int_2$ from (\ref{Int1}) and (\ref{Int2}) while making use of the nested structure of the overall integration problem on $\Productspace$. To this end, we first consider integration and quadrature on each domain $\Omega_i$, $i=1,2$, separately. Then, we combine these quadratures appropriately. In the simplest case this is done in a product-type fashion, e.g. in a benefit-cost balanced way as already presented in \cite{GriebelOettershagen2019}, which leads to a so-called full tensor product (FTP) rule on $\Productspace$. 
		\\
		Moreover, we can exploit a multilevel hierarchy of quadratures on each $\Omega_i$, $i=1,2$, and can build, following the sparse grid idea \cite{BungartzGriebel2004}, a sparse tensor product (STP) rule on $\Productspace$, which relates to the multilevel quadrature approach \cite{giles2015multilevel}. Indeed, such a STP-quadrature was already presented and analyzed for the plain integration and interpolation problem on $\Productspace$ in \cite{GriebelHarbrecht2011}, \cite{GriebelHarbrechtMulterer2015}, i.e. for the simple case $F(z,\th,t)=t$ and $\innerintegrand\in\mathcal{H}_1(\OM,\nu)\otimes\mathcal{H}_2(\OMM,\mu)$ \textit{without} any intermediate function $F$. 
		\\
		We will now generalize these results to the case of quite general integrable functions $F$, which are additionally H\"older continuous. This covers a wide range of practical applications from econometrics involving especially GMM and ML estimators. Thus, as building blocks, we assume to separately have on each $\Omega_i$, $i=1,2$, a sequence of $d_i$-dimensional quadrature rules
		\begin{align*}
			\Qi(f):=\sum_{n_i=1}^{\Ni}w_{n_i}^{i,l_i}f(x_{n_i}^{i,l_i})
		\end{align*}
		for some integrand $f:\Omega_i\ra\R$.
		\\
		There are various types of such quadrature rules such as MC, QMC, Frolov, sparse grid (SG) or product rules, among others. Their properties are expressed in their costs/degrees of freedom and their convergence rate, which depends on the regularity and smoothness of the respective integrand and thus on the underlying function space on $\Omega_i$, $i=1,2$. \textcolor{commentcolor}{To this end, we assume that $\Qi$ involves 
	\begin{align}\label{size QMC}
		\Ni \asymp 2^{l_i}	
	\end{align}
nodes in $d_i$ dimensions if it is a MC, QMC or Frolov rule. Here and in the following, the notation $\asymp$  will be short for the existence of $l_i$-independent constants $c,C>0$ such that 
$ c \, 2^{l_i}\leq \Ni \leq C \, 2^{l_i}$ for all $ l_i \geq 0$.
		If only the upper bound holds, we use the notation $\lesssim$. Moreover we assume that $\Qi$ involves\footnote{\textcolor{commentcolor}{The constants $c,C$ involved in the $\asymp$- or $\lesssim$-notation might depend on the dimensionality $d_i$.}}
		\begin{align} \label{size SG}
			\Ni \asymp 2^{l_i}l_i^{d_i-1}
		\end{align}
		nodes in $d_i$ dimensions if it is a sparse grid rule \cite{DuTeUl2015}}. Of course, for a direct product rule, we would have
		\textcolor{commentcolor}{
		\begin{align} \label{size product}
			\Ni\asymp 2^{d_il_i}
		\end{align}
		}
		and we encounter the curse of dimension with respect to $d_i$. 
		\\
		The error of a quadrature rule and thus its convergence rate depends on the involved number of nodes but also on the considered function class on $\Omega_i$ and the dimension $d_i$. In the following, we assume that the $\Qi$, $i=1,2$, have degree of exactness $p_i$, respectively. Then they have error bounds of the form
		\begin{align} \label{ErrorBounds Qi}
			e\left(\Qi,\SobolevspaceI\right)=O\left(\Ni^{-s_i}\log(\Ni)^{t_i}\right)
		\end{align}
		for $s_i:=\min\{r_i,p_i\}$. Here, we take the mixed Sobolev space $\SobolevspaceI(\Omega_i)$ of $r_i$-th bounded mixed derivatives as suitable function space for (Q)MC, Frolov and SG quadratures into consideration, where in general $r_i=0$ (with  $H^0_\text{mix}(\Omega_i)=L^2(\Omega_i)$) for MC quadrature, $r_i=1$ for QMC rules and $r_i\geq 1$ for Frolov and SG quadratures is assumed. The $\log$-exponent in (\ref{ErrorBounds Qi}) is of the general form $t_i=t_i(s_i,d_i)$ to include all of the previously introduced rules. Alternatively, for the isotropic Sobolev case and the product rule, we assume that
		\begin{align} \label{ErrorBounds Qi iso}
			e\left(\Qi,H^{r_i}\right)=O\left(\Ni^{-s_i/d_i}\right)\,.
		\end{align}
		This way, the order of convergence $s_i$ incorporates both the smoothness $r_i$ of the considered integrands and the maximal degree of exactness $p_i$ of $\Qi$ into the bounds (\ref{ErrorBounds Qi}) and (\ref{ErrorBounds Qi iso}). In particular, given an $r_i$-times differentiable integrand, one would like to choose a quadrature rule with degree of exactness at least $r_i$ in order to maximize the order of convergence. For the product rule for functions in the isotropic Sobolev space, the additional division by $d_i$ induces the curse of dimensionality, where the order of convergence decreases drastically for higher dimensional integrals.
		\\
		Now, we define difference quadrature formulae by
		\begin{align*}
			\DQuadI_{l_1}(F_{\innerintegrand})&:=
			\begin{cases}
			\QuadI_{l_1}(F_{\innerintegrand})-\QuadI_{l_1-1}(F_{\innerintegrand})\,,&\text{ for }l_1\geq 2\,,\\
			\QuadI_1(F_{\innerintegrand})\,,&\text{ for }l_1=1\,,
			\end{cases}\\
			\DQuadII_{l_2}(\innerintegrand,z)&:=
			\begin{cases}
			F(z,\th,\QuadII_{l_2}(\innerintegrand,z))-F(z,\th,\QuadII_{l_2-1}(\innerintegrand,z))\,,&\text{ for }l_2\geq 2\,,\\
			F(z,\th,\QuadII_1(\innerintegrand,z))\,,&\text{ for }l_2=1\,.
			\end{cases}
		\end{align*}
		This allows for telescopic expansions of $\QuadI_{l_1}$ and $F(z,\th,\QuadII_{l_2}(\innerintegrand,z))$ for any $z\in\OMM$. In particular, we can sum up $\DQuad_{l_i}$ over $l_i\in\N$, $i=1,2$, and get series representations of $\Int_1$ and $F_{\innerintegrand}(z)$ resulting in
		\begin{align}
			\Int_1(F_{\innerintegrand})
			=\sum_{l_1=1}^{\infty}\DQuadI_{l_1}(F_{\innerintegrand})=\sum_{l_1=1}^{\infty}\DQuadI_{l_1}\left(\sum_{l_2=1}^{\infty}\DQuadII_{l_2}(\innerintegrand,\cdot)\right) %\nonumber \\
			=\sum_{(l_1,l_2)\in\N^2}\DQuadI_{l_1}\circ \DQuadII_{l_2}(\innerintegrand,\cdot)\,.\label{integral-expansion}
		\end{align}
		Here, the ``$\cdot$" serves as placeholder for the quadrature nodes $z_{m,l_1}$ defined by each difference quadrature rule $\DQuadI_{l_1}$, \textcolor{commentcolor}{and $\DQuadI_{l_1}\circ \DQuadII_{l_2}$ is the concatenation of both operators \textcolor{commentcolor}{which is in general not commutative}. Note here that this allows us later to estimate $|\DQuadI_{l_1}\circ \DQuadII_{l_2}(\innerintegrand,\cdot)|\leq||\DQuadI_{l_1}||\cdot|\DQuadII_{l_2}(\innerintegrand,\cdot)|$}.
		\\
		For a general level index set $\mathcal{A}\subset\N^2$, we then obtain the \textit{general sparse grid quadrature rule} $Q_\mathcal{A}$ (STP) on $\Productspace$ by properly truncating the above sum, i.e.
		\begin{align} \label{GeneralSG}
			\Int_1(F_{\innerintegrand})\approx Q_\mathcal{A}(F_{\innerintegrand}):=\sum_{(l_1,l_2)\in\mathcal{A}}\DQuadI_{l_1}\circ \DQuadII_{l_2}(\innerintegrand,\cdot)\,.
		\end{align}
		For our considerations, we use the basic anisotropic SG index set 
		\begin{align*}
			\LSSG(L):=\{(l_1,l_2)\ :\ \sigma l_1+\frac{l_2}{\sigma}\leq L\}
		\end{align*}
		and compare it to the basic anisotropic full grid set 
		\begin{align*}
			\LSFG(L):=\{(l_1,l_2)\ :\ \max\{\sigma l_1,\frac{l_2}{\sigma}\}\leq L\}\,.
		\end{align*}
		Here, the parameter $\sigma>0$ accounts for different convergence rates of the inner and the outer quadrature and ``balances" them properly.
		\\
		We write $\QFG$ for the level-$L$-FTP rule with index set $\LSFG(L)$ in (\ref{GeneralSG}) and we write $\QSG$ for the level-$L$-STP rule for the index set $\LSSG(L)$ in (\ref{GeneralSG}) and define the corresponding errors as
		\begin{align}
			\EFG(F_{\innerintegrand})&:=|\Int_1(F_{\innerintegrand})-\QFG(F_{\innerintegrand})|\,,\\
			\ESG(F_{\innerintegrand})&:=|\Int_1(F_{\innerintegrand})-\QSG(F_{\innerintegrand})|\,.
		\end{align}
		First, we count the number of nodes in $\QFG$ and $\QSG$. Since $F$ does not affect the number $\NFG$ and $\NSG$ of associated nodes we can adapt Theorem 4.1 from \cite{GriebelHarbrecht2011} and adjust it for the additional case where the single quadrature rules $\Qi$, $i=1,2$, might be SG rules themselves.
		\begin{theorem}\label{TPsize} (Size of full and sparse tensor product quadrature)\\
			Let $\Qi$ on $\Omega_i$, $i=1,2$, be MC, QMC, Frolov or SG quadrature rules (i.e. with number of nodes as in (\ref{size QMC}) or (\ref{size SG}) and error bound as in (\ref{ErrorBounds Qi}). Then, the full tensor product rule $\QFG$ has $\NFG$ nodes where
			\textcolor{commentcolor}{
			\begin{align*}
				\NFG \asymp 2^{L(\sigma+\frac{1}{\sigma})}L^{q_1+q_2}\,,
			\end{align*}
			where $q_i=d_i-1$ for sparse grid rules and $q_i=0$ otherwise.}
			\\
			The sparse tensor product rule $\QSG$ has $\NSG$ nodes \textcolor{commentcolor}{where
			\begin{align*}
				\NSG\lesssim
				\begin{cases}
					L^{q_1+q_2} 2^{\max\{\sigma,1/\sigma\} L} &\text{ for }\sigma\neq 1\,,\\
					L^{q_1+q_2+1}2^{L} &\text{ for }\sigma=1\,,\\
				\end{cases}
			\end{align*}
			where again $q_i=d_i-1$ for sparse grid rules and $q_i=0$ otherwise.
			}
		\end{theorem}
		\begin{proof}
			{\color{commentcolor} First, we note that it holds for MC, QMC, Frolov and Sparse Grid rules that 
			\begin{align} \label{eqn_numnodes2}
				\Ni \asymp 2^{l_i} l_i^{q_i}\,,
			\end{align}
			where $q_i=d_i-1$ for sparse grid rules and $q_i=0$ otherwise. The estimate for the full grid $\NFG$ follows directly from expanding $\QFG$:
			\begin{align} \label{FTPExpansion}
				\QFG(F_{\innerintegrand})
				=\sum_{\max\{\sigma l_1,\frac{l_2}{\sigma}\}\leq L}\DQuadI_{l_1}\circ \DQuadII_{l_2}(\innerintegrand,\cdot)
				=\QuadI_{L/\sigma}\left(F\left(\cdot,\th,\QuadII_{\sigma\cdot L}(\innerintegrand,\cdot)\right)\right)\,.
			\end{align}
			Then we apply \eqref{eqn_numnodes2} and obtain 
			\begin{align*}
				\NFG \asymp 2^{\frac{L}{\sigma}} \left( \frac{L}{\sigma} \right)^{q_1} \cdot 2^{\sigma L} \left( \sigma L \right)^{q_2}  \asymp 2^{L(\sigma+\frac{1}{\sigma})}L^{q_1+q_2}\,.
			\end{align*}
			For the sparse grid case $\QSG$, we consider each term $\DQuadI_{l_1}\circ \DQuadII_{l_2}$ and again use \eqref{eqn_numnodes2} to compute
			\begin{align}\label{UpperBoundNSGDerivation}
				\NSG&\asymp \sum_{\sigma l_1+\frac{l_2}{\sigma}\leq L} 2^{l_1} l_1^{q_1} \cdot 2^{l_2} l_2^{q_2} 
				=\sum_{l_1=0}^{L/\sigma}\sum_{l_2=0}^{\sigma L-\sigma^2 l_1} 2^{l_1} 2^{l_2} l_1^{q_1} l_2^{q_2} \\
				&\lesssim L^{q_1+q_2} \sum_{l_1=0}^{L/\sigma} 2^{l_1} \sum_{l_2=0}^{\sigma L-\sigma^2 l_1} 2^{l_2} \asymp L^{q_1+q_2}2^{\sigma L} \sum_{l_1=0}^{L/\sigma} 2^{l_1(1-\sigma^2)}\nonumber\,.
			\end{align}
			For $\sigma>1$, this implies $\NSG \lesssim L^{q_1+q_2} 2^{\sigma L}$ and for $\sigma<1$ we get $\NSG \lesssim L^{q_1+q_2} 2^{L/\sigma}$. Only for $\sigma=1$ we have $\NSG \lesssim L^{q_1+q_2+1}2^{L}$.
			}
		\end{proof}
		\textcolor{commentcolor}{
		\begin{remark}\label{LowerBoundNSG}
			Note that we can also provide a lower asymptotic bound on $\NSG$ which will be useful later on. Starting similarly as in (\ref{UpperBoundNSGDerivation}), we get
			\begin{align*}
				\NSG\asymp\sum_{l_1=0}^{L/\sigma}\sum_{l_2=0}^{\sigma L-\sigma^2 l_1} 2^{l_1} 2^{l_2} l_1^{q_1} l_2^{q_2}\gtrsim\sum_{l_1=0}^{L/\sigma}\sum_{l_2=0}^{\sigma L-\sigma^2 l_1} 2^{l_1} 2^{l_2}\gtrsim  2^{L/\sigma}
			\end{align*}
			after leaving out all but one summand.
		\end{remark}
		}
		We observe that the reduction from FTP to STP is most substantial if $\sigma$ is close to 1, since the factor $\sigma+1/\sigma$ is reduced to $\max\{\sigma,1/\sigma\}$.
		\\To prove bounds for $\EFG$ and $\ESG$ similar to those in \cite{GriebelHarbrecht2011} but generalized with respect to the intermediate function $F$, we need the well-known notion of \textit{H\"older continuity}.
		\begin{definition}\label{Hoelder} (H\"older continuity)\\
			A function $f:\Omega\subset\R^d\ra\R$ is called \textit{H\"older continuous} if there exist $\alpha,C>0$ s.t.
			\begin{align*}
				|f(x)-f(y)|\leq C||x-y||^\alpha
			\end{align*}
			for all $x,y\in\Omega$.
		\end{definition}
		We then have the following result:
		\begin{theorem}\label{TPconvergence} (Error bound for full tensor product quadrature)\\
			Let $\Qi$, $i=1,2$, be quadrature formulas as in Theorem \ref{TPsize} and let $s_i$ be their order of convergence, respectively. Suppose that $F(\cdot,\th,\cdot)\in H_{\text{mix}}^{r_1}(\OM\times\R;\nu)$ for every $\th\in\Theta$, $\Int_2(\innerintegrand,\cdot)\in H_{\text{mix}}^{r_1}(\OM;\nu)$ and $F(z,\th,\cdot)$ is H\"older continuous with exponent $\alpha$ for any $z\in\OM$ and for any $\th\in\Theta$ and $\innerintegrand(z,\cdot)\in H_{\text{mix}}^{r_2}(\OMM,\nu)$ \textcolor{commentcolor}{with a uniform bound on $||\innerintegrand(z, \cdot)||$} for all $z\in\OM$. Then, \textcolor{commentcolor}{with $\tilde{t}_i:=t_i-s_iq_i$,} the error of the full tensor product quadrature rule $\QFG$ is given by
			\textcolor{commentcolor}{
			\begin{align*}
				\EFG=O\left((L\sigma)^{\alpha \tilde{t}_2}\left(\frac{L}{\sigma}\right)^{\tilde{t}_1}2^{-L\min(s_1/\sigma,\sigma\alpha s_2)}\right)\,.
			\end{align*}
			}
		\end{theorem}
		\begin{proof}
			We reuse the expansion (\ref{FTPExpansion}) and omit for simplicity the dependence on $\th$ in the following. With the triangle inequality and (\ref{ErrorBounds Qi}) we get
			\begin{align*}
				\EFG(F_{\innerintegrand})
				&=\left|\Int_1(F_{\innerintegrand})-\QFG(F_{\innerintegrand})\right|\\
				&=\left|\Int_1\left(F\left(\cdot,\Int_2(\innerintegrand,\cdot)\right)\right)-\QuadI_{L/\sigma}\left(F\left(\cdot,\QuadII_{\sigma\cdot L}(\innerintegrand,\cdot)\right)\right)\right|\\
				&\leq\left|\left(\Int_1-\QuadI_{L/\sigma}\right)\left(F\left(\cdot,\Int_2(\innerintegrand,\cdot)\right)\right)\right|+\\
				\ \ \ \ \ &\left|\QuadI_{L/\sigma}\left(F\left(\cdot,\Int_2(\innerintegrand,\cdot)\right)-F\left(\cdot,\QuadII_{\sigma\cdot L}(\innerintegrand,\cdot)\right)\right)\right| \,.
			\end{align*}
			The first term measures the approximation accuracy of $\QuadI_{l_1}$. Since $F(\cdot,\th,\cdot)\in H_{\text{mix}}^{r_1}(\OM)$ for any $\th\in\Theta$ and $\Int_2(\innerintegrand,\cdot)\in H_{\text{mix}}^{r_1}(\OM)$ and since $s_1$ is the order of convergence for $\QuadI_{l_1}$, we obtain for it the rate (\ref{ErrorBounds Qi}) with $i=1$, i.e.
			\begin{align*}
				\left|\left(\Int_1-\QuadI_{L/\sigma}\right)\left(F\left(\cdot,\Int_2(\innerintegrand,\cdot)\right)\right)\right|
				\lesssim N_{1,L/\sigma}^{-s_1}\log(N_{1,L/\sigma})^{t_1}\,.
			\end{align*}
			For the second term, we use H\"older continuity, the fact that $\innerintegrand(z,\cdot)\in H_{\text{mix}}^{r_2}(\OMM)$ for all $z$ and the boundedness of the operator $\QuadI_{L/\sigma}$, \textcolor{commentcolor}{which follows from the fact that point evaluation is a bounded functional in  $H_{\text{mix}}^{s_i}$, as this space is a reproducing kernel Hilbert space for all $s_i>1/2$ \cite{berlinet}}. Thus, we obtain also for $\QuadII_{l_2}$ the rate (\ref{ErrorBounds Qi}) for $i=2$, i.e.
			\begin{align*}
				\Big|\QuadI_{L/\sigma} \Big(F\left(\cdot,\Int_2(\innerintegrand,\cdot)\right)&-F\left(\cdot,\QuadII_{\sigma\cdot L}(\innerintegrand,\cdot)\right)\Big)\Big|\\
				& \leq  ||\QuadI_{L/\sigma}||\left|F\left(\cdot,\Int_2(\innerintegrand,\cdot)\right)-F\left(\cdot,\QuadII_{\sigma\cdot L}(\innerintegrand,\cdot)\right)\right|\\
				& \leq C\left|\Int_2(\innerintegrand,\cdot)-\QuadII_{\sigma\cdot L}(\innerintegrand,\cdot)\right|^\alpha\\
				&\lesssim N_{2,\sigma\cdot L}^{-\alpha s_2}\log(N_{2,\sigma\cdot L})^{\alpha t_2}\,.
			\end{align*}
			Combining both summands yields the desired result
			\begin{align*}
				\EFG(F_{\innerintegrand})&\lesssim N_{1,L/\sigma}^{-s_1}\log(N_{1,L/\sigma})^{t_1}+N_{2,\sigma\cdot L}^{-\alpha s_2}\log(N_{2,\sigma\cdot L})^{\alpha t_2}\\
				&\lesssim 2^{-s_1L/\sigma}(L/\sigma)^{t_1-s_1 q_1}+2^{-\alpha s_2L\sigma}(L\sigma)^{\alpha t_2-\alpha s_2q_2}\\
				&=O\left((L\sigma)^{\alpha \tilde{t}_2}\left(\frac{L}{\sigma}\right)^{\tilde{t}_1}2^{-L\min(s_1/\sigma,\alpha s_2\sigma)}\right)\,,
			\end{align*}
			\textcolor{commentcolor}{where $\tilde{t}_i=t_i-s_iq_i$. Note that $N_{1,L/\sigma}$ and $N_{2,\sigma\cdot L}$ have the same asymptotic bound (\ref{eqn_numnodes2}) in $L$, which provides the desired result by using the calculation 
			\begin{align}\label{TrafoNtol}
				\Ni^{-s_i}\log(\Ni)^{t_i}&\asymp 2^{-s_il_i}l_i^{-s_iq_i}\log\left(2^{l_i}l_i^{q_i}\right)^{t_i}\nonumber\\
				&=2^{-s_il_i}l_i^{-s_iq_i}\left(\log(2)l_i+q_i\log(l_i)\right)^{t_i}\nonumber\\
				&\lesssim 2^{-s_il_i}l_i^{t_i-s_iq_i}\,.
			\end{align}}
		\end{proof}
		The following extension of Theorem 4.3 in \cite{GriebelHarbrecht2011} shows that STP quadrature gives a similar result.
		\begin{theorem}\label{STPconvergence} (Error bound for sparse tensor product quadrature)\\
			Let $\Quad_{l_i}$, $s_i$, for $i=1,2$, $F$ and $\innerintegrand$ be as in Theorem \ref{TPconvergence}. Then, \textcolor{commentcolor}{with $\tilde{t}_i:=t_i-s_iq_i$,} 
			the error of the sparse tensor product quadrature is given by
			\begin{align*}
				\ESG=
				\begin{cases}
					O\left((L\sigma)^{\alpha \tilde{t}_2}\left(\frac{L}{\sigma}\right)^{\tilde{t}_1}2^{-L\min\{s_1/\sigma,\alpha s_2\sigma\}}\right) &\text{ for }\frac{s_1}{\sigma}\neq\alpha s_2\sigma\,,\\
					\\
					O\left((L\sigma)^{\alpha \tilde{t}_2}\left(\frac{L}{\sigma}\right)^{\tilde{t}_1+1}2^{-L\alpha s_2\sigma}\right) &\text{ for }\frac{s_1}{\sigma}=\alpha s_2\sigma\,.
				\end{cases}
			\end{align*}
		\end{theorem}
		\begin{proof}
			Using the triangle inequality and (\ref{ErrorBounds Qi}) for $i=1$ we get a bound for $\DQuadI_{l_1}$ for functions $f\in H_{\text{mix}}^{r_1}$,
			\begin{align*}
				||\DQuadI_{l_1}||&=\max_{f\in H_{\text{mix}}^{r_1}\,,||f||\leq 1}\left|\left|\DQuadI_{l_1}(f)\right|\right|\\
				&\leq\max_{f\in H_{\text{mix}}^{r_1}\,,||f||\leq 1}\left|\left|\QuadI_{l_1}(f)-\Int_1(f)\right|\right|+\left|\left|\Int_1(f)-\QuadI_{l_1-1}(f)\right|\right|\\
				&\lesssim N_{1,l_1}^{-s_1}\log(N_{1,l_1})^{t_1}\,.
			\end{align*}
			\textcolor{commentcolor}{Analogously, and using that $F$ is H\"older continuous, we get 
			\begin{align*}
				\left|\DQuadII_{l_2}(\innerintegrand,\cdot)\right|\leq||\DQuadII_{l_2}||\cdot||\innerintegrand||\lesssim N_{2,l_2}^{-\alpha s_2}\log(N_{2,l_2})^{\alpha t_2}\,.
			\end{align*}
			}Expanding $\Int_1(F_{\innerintegrand})$ according to (\ref{integral-expansion}) \textcolor{commentcolor}{and plugging in the above estimates for $||\DQuadI_{l_1}||$ and $||\DQuadII_{l_2}||$, as well as (\ref{eqn_numnodes2}) and (\ref{TrafoNtol})}, we obtain
			\begin{align*}
				\ESG(F_{\innerintegrand})
				&=\left|\Int_1(F_{\innerintegrand})-\!\!\!\!\!\sum_{\sigma l_1+\frac{l_2}{\sigma}\leq L}\DQuadI_{l_1}\left( \DQuadII_{l_2}(\innerintegrand,\cdot)\right)\right|\\
				&\leq\sum_{\sigma l_1+\frac{l_2}{\sigma}>L}\left|\DQuadI_{l_1}\left( \DQuadII_{l_2}(\innerintegrand,\cdot)\right)\right|\\
				&\leq\sum_{\sigma l_1+\frac{l_2}{\sigma}>L} \left|\left|\DQuadI_{l_1}\right|\right|\left|\DQuadII_{l_2}(\innerintegrand,\cdot)\right|\\
				&\lesssim \sum_{\sigma l_1+\frac{l_2}{\sigma}> L} N_{1,l_1}^{-s_1}\log(N_{1,l_1})^{t_1} \cdot N_{2,l_2}^{-\alpha s_2}\log(N_{2,l_2})^{\alpha t_2} \\
				&\lesssim \sum_{\sigma l_1+\frac{l_2}{\sigma}> L}2^{-s_1l_1}l_1^{t_1-s_1q_1}2^{-\alpha s_2l_2}l_2^{\alpha t_2-\alpha s_2q_2}\\
				&\lesssim\sum_{\sigma l_1+\frac{l_2}{\sigma}>L}2^{-(s_1l_1+\alpha s_2l_2)}l_1^{\tilde{t}_1}l_2^{\alpha \tilde{t}_2}\,.
			\end{align*}
			We now split the index set $\left\{(l_1,l_2):\sigma l_1+\frac{l_2}{\sigma}>L\right\}$ into the two disjoint sets 
			\begin{align*}
				S_1:=\left\{(l_1,l_2):0\leq l_1\leq\frac{L}{\sigma}, M<l_2\right\} \quad \mbox{and} \quad
				S_2:=\left\{(l_1,l_2):\frac{L}{\sigma}<l_1, 0\leq l_2\right\}\,,
			\end{align*}
			where again $M=L\sigma-l_1\sigma^2$, and sum over each index set separately. This gives
			\begin{align*}
				\sum_{(l_1,l_2)\in S_1}
				&2^{-(s_1l_1+\alpha s_2l_2)}l_1^{\tilde{t}_1}l_2^{\alpha \tilde{t}_2}\\
				&=\sum_{l_1=0}^{L/\sigma}l_1^{\tilde{t}_1}2^{-s_1l_1}\sum_{l_2=M+1}^{\infty}l_2^{\alpha \tilde{t}_2}2^{-\alpha s_2l_2}\\
				&\leq\sum_{l_1=0}^{L/\sigma}l_1^{\tilde{t}_1}(M+1)^{\alpha \tilde{t}_2}2^{-(s_1l_1+\alpha s_2M)}\sum_{l_2=1}^{\infty}l_2^{\alpha \tilde{t}_2}2^{-\alpha s_2l_2}\\
				&\leq (L\sigma)^{\alpha \tilde{t}_2}\left(\frac{L}{\sigma}\right)^{\tilde{t}_1}\polylog_{-\alpha \tilde{t}_2}(2^{-\alpha s_2})\sum_{l_1=0}^{L/\sigma}2^{-(s_1l_1+\alpha s_2M)}\\
				&\lesssim\ (L\sigma)^{\alpha \tilde{t}_2}\left(\frac{L}{\sigma}\right)^{\tilde{t}_1}2^{-\alpha s_2L\sigma}\sum_{l_1=0}^{L/\sigma}2^{-l_1\sigma(s_1/\sigma-\alpha s_2\sigma)}\,,
			\end{align*}
			where $\polylog$ denotes the \textit{Polylogarithm}. In the same fashion we obtain
			\begin{align*}
				\sum_{(l_1,l_2)\in S_2}2^{-(s_1l_1+\alpha s_2l_2)}
				\lesssim\left(\frac{L}{\sigma}\right)^{\tilde{t}_1}2^{-s_1L/\sigma}
				\lesssim \left(\frac{L}{\sigma}\right)^{\tilde{t}_1}2^{-\alpha s_2L\sigma}\ 2^{-L(s_1/\sigma-\alpha s_2\sigma)}\,.
			\end{align*}
			Joining both sums we distinguish three cases: For $\frac{s_1}{\sigma}<\alpha s_2\sigma$, we have 
			\begin{align*}
				\sum_{\sigma l_1+\frac{l_2}{\sigma}>L}
				&2^{-(s_1l_1+\alpha s_2l_2)}l_1^{\tilde{t}_1}l_2^{\alpha \tilde{t}_2}\\
				\lesssim&\ 2^{-\alpha s_2L\sigma}\left(\frac{L}{\sigma}\right)^{\tilde{t}_1}\times\\
				&\left((L\sigma)^{\alpha \tilde{t}_2}\sum_{l_1=0}^{L/\sigma}2^{-l_1\sigma(s_1/\sigma-\alpha s_2\sigma)}+2^{-L(s_1/\sigma-\alpha s_2\sigma)}\right)\\
				\lesssim&\ 2^{-\alpha s_2L\sigma}\left(\frac{L}{\sigma}\right)^{\tilde{t}_1}\times\\
				&\left((L\sigma)^{\alpha \tilde{t}_2}2^{-L(s_1/\sigma-\alpha s_2\sigma)}+2^{-L(s_1/\sigma-\alpha s_2\sigma)}\right)\\
				\lesssim&\ (L\sigma)^{\alpha \tilde{t}_2}\left(\frac{L}{\sigma}\right)^{\tilde{t}_1}2^{-s_1L/\sigma}=O\left((L\sigma)^{\alpha \tilde{t}_2}\left(\frac{L}{\sigma}\right)^{\tilde{t}_1}2^{-s_1L/\sigma}\right)\,.
			\end{align*}
			For $\frac{s_1}{\sigma}>\alpha s_2\sigma$, we can bound the expression by
			\begin{align*}
				\sum_{\sigma l_1+\frac{l_2}{\sigma}>L}2^{-(s_1l_1+\alpha s_2l_2)}l_1^{\tilde{t}_1}l_2^{\alpha \tilde{t}_2}
				&\lesssim\ 2^{-\alpha s_2L\sigma}\left(\frac{L}{\sigma}\right)^{\tilde{t}_1}\left((L\sigma)^{\alpha \tilde{t}_2}\!\!\!+1\right)\\
				&\lesssim\ (L\sigma)^{\alpha \tilde{t}_2}\left(\frac{L}{\sigma}\right)^{\tilde{t}_1}2^{-\alpha s_2L\sigma}\\
				&=O\left((L\sigma)^{\alpha \tilde{t}_2}\left(\frac{L}{\sigma}\right)^{\tilde{t}_1}2^{-\alpha s_2L\sigma}\right)\,.
			\end{align*}
			Finally, for $\frac{s_1}{\sigma}=\alpha s_2\sigma$, we have
			\begin{align*}
				\sum_{\sigma l_1+\frac{l_2}{\sigma}>L}2^{-(s_1l_1+\alpha s_2l_2)}l_1^{\tilde{t}_1}l_2^{\alpha \tilde{t}_2}&\lesssim\ 2^{-\alpha s_2L\sigma}\left(\frac{L}{\sigma}\right)^{\tilde{t}_1}\left((L\sigma)^{\alpha \tilde{t}_2}\sum_{l_1=0}^{L/\sigma}1+1\right)\\
				&\lesssim\ (L\sigma)^{\alpha \tilde{t}_2}\left(\frac{L}{\sigma}\right)^{\tilde{t}_1+1}2^{-L\alpha s_2\sigma}\\
				&=O\left((L\sigma)^{\alpha \tilde{t}_2}\left(\frac{L}{\sigma}\right)^{\tilde{t}_1+1}2^{-L\alpha s_2\sigma}\right)\,.
			\end{align*}
			This concludes the proof.
		\end{proof}
		\bigbreak 
		Theorems \ref{TPconvergence} and \ref{STPconvergence} can also be stated for probabilistic error rates of the outer quadrature, e.g. from MC integration. Then, the mean squared error is used instead of a norm and the proof proceeds similar to the derivation of the mean squared error of MC integration.
		\\
		For both, $\QFG$ and $\QSG$, we can combine Theorems \ref{TPsize} and \ref{TPconvergence} or \ref{STPconvergence}, respectively, to obtain an error bound in terms of the costs of the corresponding quadrature formula.
		\begin{corollary}\label{corollary}
			Let $\Quad_{l_i}$, $s_i$, for $i=1,2$, $F$ and $\innerintegrand$ be as in Theorem \ref{TPconvergence} and set
			\textcolor{commentcolor}{
			\begin{align*}
				\gamma_\infty:=\frac{\min\{s_1/\sigma,\alpha s_2\sigma\}}{\sigma+1/\sigma}>0\quad \mbox{and}\quad 
				\gamma_1:=\frac{\min\{s_1/\sigma,\alpha s_2\sigma\}}{\max\{\sigma,1/\sigma\}}>0\,.
			\end{align*}
			Furthermore, define $\tilde{t}_1^w:=\tilde{t}_1+\gamma_wq_1$ and $\tilde{t}_2^w:=\tilde{t}_2+\frac{\gamma_wq_2}{\alpha}$ for $w\in\{1,\infty\}$. 
			\\
			The FTP error is 
			\begin{align*}
				\EFG=O\left(\log(\NFG)^{\tilde{t}_1^\infty+\alpha \tilde{t}_2^\infty}(\NFG)^{-\gamma_\infty}\right)\,.
			\end{align*}
			For STP quadrature we distinguish four cases according to the Theorems \ref{TPsize} and \ref{STPconvergence}. If $\sigma\neq1$ we have 
			\begin{align*}
				\ESG=
				\begin{cases}
					O\left((\log \NSG)^{\tilde{t}_1^1+\alpha \tilde{t}_2^1}(\NSG)^{-\gamma_1}\right) &\text{ for }\frac{s_1}{\sigma}\neq\alpha s_2\sigma\,,\\
					\\
					O\left((\log \NSG)^{\tilde{t}_1^1+\alpha \tilde{t}_2^1+1}(\NSG)^{-\gamma_1}\right) &\text{ for }\frac{s_1}{\sigma}=\alpha s_2\sigma\,,
				\end{cases}
			\end{align*}
			and if $\sigma=1$
			\begin{align*}
				\ESG=
				\begin{cases}
					O\left((\log \NSG)^{\tilde{t}_1^1+\alpha \tilde{t}_2^1+\gamma_1}(\NSG)^{-\gamma_1}\right) &\text{ for }s_1\neq\alpha s_2\,,\\
					\\
					O\left((\log \NSG)^{\tilde{t}_1^1+\alpha \tilde{t}_2^1+\gamma_1+1}(\NSG)^{-\gamma_1}\right) &\text{ for }s_1=\alpha s_2\,.
				\end{cases}
			\end{align*}
			}
		\end{corollary}
		\begin{proof}
			\textcolor{commentcolor}{
			For the estimate for $\EFG$ we use a similar trick as Harbrecht and Griebel \cite{GriebelHarbrecht2011}: The estimate 
			\begin{align*}
				\NFG \asymp 2^{L(\sigma+\frac{1}{\sigma})}L^{q_1+q_2}
			\end{align*}
			from Theorem \ref{TPsize} implies 
			\begin{align*}
				L\asymp\log(\NFG)-(q_1+q_2)\log(L)\lesssim\log(\NFG)
			\end{align*}
			and hence 
			\begin{align*}
				&\frac{\NFG}{\log(\NFG)^{q_1+q_2}}\lesssim 2^{L(\sigma+1/\sigma)}\\
				\Leftrightarrow 2^{-L\min(s_1/\sigma,\sigma\alpha s_2)}=&\left(2^{L(\sigma+1/\sigma)}\right)^{-\gamma_\infty}\lesssim\left(\frac{\NFG}{\log(\NFG)^{q_1+q_2}}\right)^{-\gamma_\infty}\,.
			\end{align*}
			Plugging these estimates into the result from Theorem \ref{TPconvergence} gives the desired estimate for $\EFG$. Analogously we obtain the results for $\ESG$ from the Theorems \ref{TPsize} and \ref{STPconvergence} and the Remark \ref{LowerBoundNSG}.}
		\end{proof}
		Finally, we identify the optimal $\sigma$ to balance error bounds of $\QuadI_{l_1}$ and $\QuadII_{l_2}$ and get an optimal joint convergence rate.
		\begin{theorem} (Optimal $\sigma$ for full and sparse tensor product quadrature)\\
			Both, $\QFG$ and $\QSG$, achieve their best error bound for
			\begin{align} \label{OptimalSigma}
				\sigma^*=\sqrt{\frac{s_1}{\alpha s_2}}\,.
			\end{align}
			If $\kappa:=\frac{s_1}{\alpha s_2}\neq 1$, then any $\sigma$ with $\sigma^2\in[1,\kappa]$ or $\sigma^2\in[\kappa,1]$, respectively, is optimal for $\QSG$. The optimal exponents are then
			\begin{align}
				\gamma_\infty^*=\frac{\alpha s_1s_2}{s_1+\alpha s_2}\quad \mbox{and} \quad
				\gamma_1^*=\min\{s_1,\alpha s_2\}\,.
			\end{align}
		\end{theorem}
		\begin{proof}
			In order to achieve optimal bounds, we have to maximize $\gamma_\infty$ and $\gamma_1$. Here $\gamma_\infty$ is maximized if $s_1/\sigma=\alpha s_2\sigma$, i.e. $\sigma^2=\kappa$. For $\gamma_1$, we have
			\begin{align} \label{optimal gamma formula}
				\gamma_1^*:=\max_{\sigma>0}\gamma_1=\max_{\sigma>0}\left(\alpha s_2\min\{\kappa,\sigma^2\}\min\{1,\frac{1}{\sigma^2}\}\right)\,.
			\end{align}
			For $\kappa<1$, i.e. $s_1<\alpha s_2$, we distinguish the cases (I) $\sigma^2<\kappa$, (II) $\kappa\leq\sigma^2\leq1$ and (III) $\sigma^2>1$ and have $\gamma_1^*=s_1$ for (II) and $\gamma_1^*<s_1$ for (I) and (III). Similar cases result from $\kappa>1$ with $\gamma_1^*$ maximal for $1\leq\sigma^2\leq\kappa$. For $\kappa=1$, i.e. $s_1=\alpha s_2$, we get $\gamma_1=s_1$ since (\ref{optimal gamma formula}) is then maximized by $\sigma=1$.
		\end{proof}
		All of the above theorems as well as Corollary \ref{corollary} can be easily adjusted to cover the usage of the product rule for one or both domains and to provide corresponding results, i.e. we can extend them to functions in isotropic Sobolev spaces. \textcolor{commentcolor}{F}or example, in the case of product rules for both domains we have for FTP quadrature
		\begin{align*}
			\NFG=O\left(2^{L(d_1/\sigma+d_2\sigma)}\right)\,,
		\end{align*}
		and for STP quadrature
		\begin{align*}
			\NSG=
			\begin{cases}
				O\left(2^{Ld_2\sigma}\right) &\text{ for }\sigma^2<d_1/d_2\,,\\
				O\left(2^{Ld_1/\sigma}\right) &\text{ for }\sigma^2>d_1/d_2\,,\\
				O\left(\frac{L}{\sigma}2^{Ld_2\sigma}\right) &\text{ for }\sigma^2=d_1/d_2\,.\\
			\end{cases}
		\end{align*}
		Moreover, we have similar exponents when we combine the product rule on one domain with one of the other previously mentioned rules on the second domain and vice-versa. This means that the size of both, FTP and STP quadrature, depends on the dimension of one or both domains if the product rule is involved, and therefore the curse of dimensionality transfers to the double quadrature. However, the STP quadrature still improves the FTP quadrature by reducing the number of nodes significantly.
		\\
		Similarly, error bounds for FTP and STP quadrature based on the product rule can be derived via the same steps as in the proofs for Theorems \ref{TPconvergence} and \ref{STPconvergence}. However, combined with the corresponding numbers of nodes of the rules, we then observe that the curse of dimensionality highly affects the cost complexities and makes the product rule prohibitively slow for high-dimensional integration domains. Hence, we abstain here from presenting the proofs for the product rule cases and instead only focus on the more promising quadrature rules on mixed Sobolev spaces. 
				
	\section{Discrete Choice Models and Estimation}\label{DCM}
		In the following, we derive the \textit{Mixed Logit} and the \textit{Multinomial Probit} model as popular specifications for Discrete Choice models (DCM) to obtain two test functions which fit to the setting described in section \ref{MLE-theory}. DCM are used to understand how individuals $n=1,...,N$ choose between alternatives $i=1,...,J$. For each alternative, we want to find a \textit{choice probability} $p(y_n=i|z_n,u)$ as a function of the observed attributes $(y_n,z_n)$ for each individual $n$ and a parameter vector $\th$, where $y_n$ is the observed decision and $z_n$ is a vector of observed attributes of individual $n$.
		\\
		Discrete Choice models have been used for many years in different branches of econometrics: Research applications include the analysis of market equilibria \cite{BLP1995}, transportation (\cite{Bhat1998}, both for Mixed Logit) or debt crisis in developing countries (\cite{hajivassiliou1994b}, for Multinomial Probit). Train \cite{Train2009} gives a comprehensive overview of various models, applications, estimation techniques and respective numerical methods.
		\\
		In terms of econometric model classes, Mixed Logit models and (Mixed) Multinomial Probit models can be described as Generalized Linear Mixed models (GLMM). For many GLMM, due to a so-called \textit{mixing distribution}, the estimation requires the computation of a multi-dimensional integral. As this integral has often no analytical solution it needs to be approximated. A survey of the optimal quadrature rules for several GLMM and their dependence on the tested parameter set can be found in \cite{Gilch2020}.
		\\
		Each individual makes his choice according to an (to the researcher) unobservable utility measure $V(z,u)\in\R^J$ and the alternative with maximal utility is chosen. Let $u\in\OMM\subset\R^q$ be a vector of parameters for the utility measure $V$ and let $z\in\OM\subset\R^{J\times q}$ be a vector (or matrix) of observed exogenous variables. Here, $\OM$ is the domain for the exogenous variable $z$ and we require that the model is identified. Furthermore, we suppose there are unobservable factors or errors $\varepsilon\in\R^J$ which affect the utility and are distributed according to a known (or assumed) distribution. The common utility function $V:\OM\times\OMM\ra\R^J$ is linear in $z$ and $u$ and additive in $\varepsilon$,
		\begin{align} \label{LinearModel}
			V(z,u)=zu+\varepsilon\,.
		\end{align}
		An individual chooses alternative $y=i$ exactly if $V_i(z,u)>V_j(z,u)$ for all $j\neq i$, where $V_i$ are the components of the vector $V$, i.e. the utility of the particular choice $i$. In order to find the choice probabilities for a given parameter vector $u$ and exogenous variable $z$, i.e.
		\begin{align}\label{ChoiceProb}
			P(y=i|z,u)=P\left((zu)_i-(zu)_j> \varepsilon_j-\varepsilon_i\ :\ \forall j\neq i\right)\,,
		\end{align}
		we need to propose a distribution for $\varepsilon$. The Mixed Logit model assumes an extreme value distribution with p.d.f. 
		\begin{align*}
			\zeta(\varepsilon_j)=e^{-\varepsilon_j}e^{-e^{-\varepsilon_j}}
		\end{align*}
		for $j=1,...,J$, while the Multinomial Probit model uses a multivariate Gaussian distribution $\Normal(0,\Sigma) $ with covariance matrix $\Sigma\in\R^{J\times J}$.
		\\
		The choice probabilities for Mixed Logit are based on the choice probabilities of the more basic Logit model which assumes a fixed parameter vector $u$ for every individual $n$: Given the error $\varepsilon_i$, we have
		\begin{align*}
			P(y=i|z,u,\varepsilon_i)
			=P\left((zu)_i-(zu)_j+\varepsilon_i> \varepsilon_j : \forall j\neq i\right)
			=\prod_{j\neq i}e^{-e^{-(\varepsilon_i+(zu)_i-(zu)_j)}}
		\end{align*}
		and then integrate over the distribution of $\varepsilon_i$ to obtain
		\begin{align*}
			P(y=i|z,u)
			&=\int_{\R} P(y=i|z,u,\varepsilon_i)e^{-\varepsilon_i}e^{-e^{-\varepsilon_i}}d\varepsilon_i =\frac{e^{(zu)_i}}{\sum_{j=1}^{J}e^{(zu)_j}}\,.
		\end{align*}
		In the standard Logit model, the goal of estimation is to find an optimal parameter vector $u$ which fits to the observed data $(y_n,z_n)$ for all individuals $n=1,..,N$. In contrast, the Mixed Logit model allows to model individuals with distinctive tastes, i.e. individual parameter vectors $u_n$. 
		\\
		Yet trying to find a complete set $\{u_n|\ n=1,...,N\}$ of parameters is computationally far too expensive. Instead, it is assumed that the $u_n$ are realizations of a random variable $U$ with density function $h(u|u_0,\Sigma)$, mean $u_0\in\R^q$ and covariance matrix $\Sigma\in\R^{q\times q}$. Then, the random taste is integrated out for each individual to get the choice probability
		\begin{align} \label{MixLogitInt}
			P(y=i|z,\th)=\int_{\R^q}\frac{e^{(zu)_i}}{\sum_{j=1}^{J}e^{(zu)_j}}h(u|u_0,\Sigma)du\,.
		\end{align}
		This constitutes the Mixed Logit choice probabilities for $i=1,...,J$, where the parameter vector $\th=(u_0,\Sigma)$ fully specifies the Mixed Logit model given a density $h$. Given assumptions on $\Sigma$, the parameter search is then reduced to $\th\in\Theta\subset\R^q$. The density is often also called mixing distribution and assumed to be a Gaussian. This leads to an integral in (\ref{MixLogitInt}), which is not analytically solvable. Other mixing distributions are possible but McCulloch and Searle \cite{mcculloch2001G-L-M-M} point out that the choice of the distribution seems to have only marginal effect on the model performance.
		
		%%%{The Multinomial Probit model}
		The Multinomial Probit model is more directly derived from (\ref{ChoiceProb}). In order to fit the model notationally to the setting of the definitions (\ref{ExtremumEstimator})-(\ref{DoubleIntegral}), we rewrite (\ref{LinearModel}) as
		\begin{align} \label{LinearModel-2}
			V(z,\th)=z\th+u\,,
		\end{align}
		so that now $u\in\OMM\subset\R^J$ is the unobservable residual or error and $\th\in\Theta\subset\R^q$ is the parameter vector to be estimated.
		\\
		Since only the differences in utility affect the choice, we can define $\tilde{V}_{ij}=V_i-V_j$, $\tilde{W}_{ij}=(z\th)_i-(z\th)_j$ and $\tilde{u}_{ij}=u_i-u_j$ for $i\neq j$ and rewrite (\ref{ChoiceProb}) (with the swapped notation $u\rightarrow\th$ and $\varepsilon\rightarrow u$) as
		\begin{align} \label{ProbitChoiceProbwithDiff}
			P(y=i|z,\th)=P\left(\tilde{V}_{ij}>0\ :\ \forall j\neq i\right)\,.
		\end{align}
		The vector $\tilde{u}_i= (\tilde{u}_{i1},...,\tilde{u}_{i(i-1)},\tilde{u}_{i(i+1)},...,\tilde{u}_{iJ})^T \in\R^{J-1}$ is normally distributed with covariance matrix $\tilde{\Sigma}_i$ derived from $\Sigma$. Then (\ref{ProbitChoiceProbwithDiff}) evaluates to
		\begin{align} \label{ProbitChoiceProb}
			P(y=i|z,\th)=\int_{\R^{J-1}}\mathbf{1}_{\{\tilde{W}_{ij}+\tilde{u}_{ij}>0\ \forall j\neq i\}}\phi(\tilde{u}_i)d\tilde{u}_i=\Phi(\tilde{W}_i)\,,
		\end{align} 
		where $\phi$ is the p.d.f. and $\Phi$ is the c.d.f. of $\tilde{u}_i$. The multivariate c.d.f. $\Phi$ cannot be computed analytically for non-trivial $\Sigma$ and also has to be approximated numerically. In contrast to the Mixed Logit model, this means that the integral does not stem from a mixing distribution $h$ but directly from the assumed probability distribution for the error in the utility function.
		\\
		For Panel or cluster data, the within-cluster and -series correlation within the Probit model is usually already expressed by the freely chosen covariance matrix $\Sigma$. Hence, an additional mixing distribution is rarely used in practice. It is mentioned, e.g.\ in \cite{mcculloch2001G-L-M-M} and \cite{Train2009}, as possibly beneficial but it is numerically less feasible. 
		\\
		However, \textcolor{commentcolor}{similarly} to the Mixed Logit model, we would assume for a Mixed Probit model that $\th$ is in fact different for every individual $n$ so we have draws $\th_n$ from a distribution $h$ as distinctive parameter vectors for every individual. Integrating over $\th$, we obtain the Mixed Probit choice probability
		\begin{align} \label{MixedProbit}
			P(y=i|z_n,\Sigma,\th_0,\Psi)=\int_{\Theta}P(y=i|z,\th)h(\th|\th_0,\Psi)d\th
		\end{align}
		for an individual $n$, covariance matrix $\Sigma$ of the multivariate Gaussian distribution and a mixing distribution $h$ with mean $\th_0$ and covariance matrix $\Psi$. Together with the integral (\ref{ProbitChoiceProb}) for the choice probability, (\ref{MixedProbit}) constitutes a double integral which can also be written in the form (\ref{DoubleIntegral}).
		\\
		Given a certain, parameterized econometric model and a set of observations for the associated economic situation, the researcher is now interested in finding the optimal parameter $\th$ that fits the model to the data. Extremum estimators find this parameter by maximizing an objective function $R_N$ which incorporates the model structure for an observed data set. In the context of DCM, this data set usually consists of observed features $z_n$ and decisions $y_n$ of individuals or firms.
		\\
		The most popular extremum estimator, the Maximum Likelihood estimator, is simply given by
		\begin{align} \label{MLEstimator}
			R^\text{ML}_N(\th):=G_N(\th)=\frac{1}{N}\sum_{n=1}^{N}m(z_n|\th)
		\end{align}
		where the moment function is defined as
		\begin{align} \label{ML-momentfunction}
			m(z_n|\th):=\sum_{j=1}^{J}y_{n,j}\log P(y=j|z_n,\th)\,,
		\end{align}
		with $y_{n,j}=1$ if individual $n$ chooses alternative $j$ and 0 otherwise. $G_N(\th)$ is often denoted by $\ell(\th)$ and called the \textit{likelihood} of the parameter vector $\th$ (given the observations $z_n$). It can be considered as the approximation of the expected value
		\begin{align*}
			\OG(\th):=\E_z[m(z|\th)]
		\end{align*}
		with respect to the real world data space $\OM$ and the data points $z$ within.
		\\
		Yet the computation of the likelihood requires knowledge or well justified assumptions on the distribution of the unobservable variables ($\varepsilon$ for Mixed Logit and $u$ for Multinomial Probit). If this requirement cannot be satisfied, the \textit{Generalized Method of Moments (GMM) estimator} provides a less restrictive alternative. It is derived from solving estimating equations
		\begin{align} \label{GMMequation}
			\OG(\th)=\E_z[m(z|\th)]=0
		\end{align}
		where the moment function is constructed from some orthogonality condition arising from the examined model. Then, the objective function in terms of extremum estimators (see (\ref{ExtremumEstimator})) is given by
		\begin{align*}
			R^\text{GMM}_N(\th):=-||G_N(\th)||=-\left|\left|\frac{1}{N}\sum_{n=1}^{N}m(z_n|\th)\right|\right|
		\end{align*}
		and maximized with respect to the parameter space $\Theta\ni\th$, $\Theta\subset\R^q$. Sometimes the Euclidean norm $||\cdot||$ is replaced by a $W$-norm $||\cdot||_W$ for some symmetric positive definite matrix $W$.
		\\
		If a choice probability $P_{n,j}(\th):=P(y=j|z_n,\th)$ is available, as for DCM, then the most efficient GMM estimator is derived from the Maximum Likelihood moment function (\ref{ML-momentfunction}) by taking the derivative and using the identity
		\begin{align*}
			0=\sum_{j=1}^{J}\nabla_\th P_{n,j}(\th)%=\sum_{j=1}^{J}P(y=j|z,\th)\nabla_\th\log P(y=j|z,\th)\,,
		\end{align*}
		so that
 		\begin{align} \label{GMM-momentfunction}
			m(z_n|\th)=\sum_{j=1}^{J}\left(\frac{y_{n,j}}{P_{n,j}(\th)}-1\right)\nabla_\th P_{n,j}(\th)\,,
		\end{align}
		which is a $q$-dimensional vector-valued moment function. Other moment functions are used if the model has weaker or different assumptions.
		\\
		Since both models define choice probabilities, Maximum Likelihood would be the natural choice for estimation. However, the fact that we cannot compute the choice probabilities exactly calls for the use of approximated extremum estimators, where the objective function $m$ is approximated. Then, under certain circumstances, GMM estimation might be the better option in terms of consistency of estimation (see \cite{GriebelOettershagen2019}, \cite{hajivassiliou1994classical} and \cite{McFadden1989}). In the next section, we test both models and estimators in order to obtain numerical results for several use cases.	
		
	\section{Numerical Results} \label{MLE-numerics}
		This section is devoted to the validation of the previously obtained results on the convergence order of FTP and STP quadrature. We present numerical results for a synthetic test function and for the two exemplary integrands from Maximum Likelihood and GMM estimation introduced in the previous section. Finally, we also examine the integral arising from a Mixed Probit model.
		\\
		Firstly, we have to establish the proper measure for accuracy of the developed quadrature rules. While the true value of the integral is available for the synthetic test function, we use as error measure for the other three integrals the relative error
		\begin{align*}
			\tilde{E}_L(F_{\innerintegrand}):=\frac{|Q_L(F_{\innerintegrand})-Q_{L-1}(F_{\innerintegrand})|}{|Q_L(F_{\innerintegrand})|}
		\end{align*}
		for $Q_L=\QFG$ or $Q_L=\QSG$ respectively.
				
		\begin{table}
			\begin{adjustwidth}{-4in}{-4in}  
				\begin{center}
					\renewcommand{\arraystretch}{1.3}
					\begin{tabular}{c||c|c|c}
						$\Qi=$ & MC & QMC & SG/Frolov/optimal weights \\
						\hline
						$e\left(\Qi,\mathcal{H}\right)=$ & $O\left(\Ni^{-1/2}\right)$ & $O\left(\Ni^{-1}\log(\Ni)^{d_i}\right)$ & $O\left(\Ni^{-s_i}\log(\Ni)^{t_i(s_i,d_i)}\right)$\\
						\hline
						$\mathcal{H}=$ & $L^2(\Omega_i)$ & $H^1_{\text{mix}}(\Omega_i)$ & $\SobolevspaceI$
					\end{tabular}
				\end{center}
			\end{adjustwidth}	
			\hspace{.2\textwidth}
			\caption[caption]{Asymptotic convergence rates for several multi-dimensional quadrature rules on a bounded $d_i$-dimensional integration domain. For Frolov cubature a zero boundary conditions is assumed. The exponent $t_i(s_i,d_i)$ serves as place holder for the different secondary rates of SG, Frolov and optimal weights quadrature.}
			\label{tab:ConvRatesQi}
		\end{table}
		We rely on several types of quadrature rules for multi-dimensional integrals on bounded domains which we employ as $\Qi$, $i=1,2$ (their convergence rates and the respective regularity assumptions are summarized in Table \ref{tab:ConvRatesQi} and conform with our definition in (\ref{ErrorBounds Qi})): Monte Carlo (MC) integration determines its quadrature nodes by randomly drawing samples from the integration domain and uses the uniform weight $\wni=1/\Ni$ for $n_i=1,...,\Ni$. It achieves a probabilistic rate of $O(\Ni^{-1/2})$ for $L^2$-integrable functions, independent of the dimensionality $d_i$ of the integration domain. Quasi Monte Carlo (QMC) integration is designed to recreate this independence of $d_i$ in a deterministic fashion by constructing the nodes from so-called \textit{low discrepancy series} \cite{DickKuoSloan2013}. Popular QMC rules are the Sobol- and the Halton-rules which both achieve convergence rates of $O(\Ni^{-1}\log(\Ni)^{d_i})$ for functions in $H^1_{\text{mix}}(\Omega_i)$ but in general no faster convergence for functions in $H^{r_i}_{\text{mix}}(\Omega_i)$ with $r_i>1$ can be obtained, i.e. $s_i=1$. Another approach for QMC rules are lattice rules which include the Frolov cubature method \cite{KaOeUlUl2018}. Frolov cubature achieves the rate $O(\Ni^{-r_i}\log(\Ni)^{(d_i-1)/2})$ for integrands in $\mathring{H}^{r_i}_{\text{mix}}(\Omega_i)$, i.e. for integrands with zero boundary, hence we have $s_i=r_i$ for Frolov cubature. All QMC rules use the same uniform weight $\wni=1/\Ni$ as MC integration. 
		\\
		In contrast, the so-called optimal weights cubature uses information about the integration domain and the function space of the considered integrands to compute optimal weights for a given set of nodes \cite{Oettershagen2017}. This way, the rather slow rate of MC and QMC quadrature can be improved significantly for functions in $\SobolevspaceI(\Omega_i)$. As yet, at least the upper bound $O(\Ni^{-r_i}\log(\Ni)^{r_id_i})$ for $r_i>1/2$ has been proven for the optimal weights MC case \cite{KaUlVo19}, whereas the lower bound $O(\Ni^{-r_i}\log(\Ni)^{(d_i-1)/2})$ is known for general best weighted sampling. This is a major improvement w.r.t.\ previous (Q)MC approaches and also captures the main rate $\Ni^{-r_i}$. Optimal weights cubature allows us to enhance MC quadrature when we have no option to obtain quadrature nodes in a systematic way, e.g.\ if the nodes are observations or samples from simulations. To this end, for the optimal weights, we just need to solve a linear equation system involving the reproducing kernel of the underlying Hilbert space $\SobolevspaceI(\Omega_i)$, which is to be present in our respective regularity assumption. For further details see \cite{Oettershagen2017}.
		\\
		Finally, SG quadrature creates a rule for multi-dimensional integrals from rules on one-dimensional domains by using only certain points of the tensor product of the one-dimensional rules. This way, the curse of dimensionality can be circumvented to some extent and an upper bound $O(\Ni^{-s_i}\log(\Ni)^{(d_i-1)(s_i+1)})$ of the rate  is achieved \cite{NovakRitter1996}, \cite{Wasilkowski1995}. Note at this point that this common but suboptimal upper bound was recently improved in \cite{DungUllrich2015} to $O(\Ni^{-s_i}\log(\Ni)^{(d_i-1)(s_i+1/2)})$ and it was shown that this is also a \textit{lower} bound of the quadrature error, i.e. it is thus the optimal rate of the SG approach. Therefore, compared to optimal weights MC quadrature, SG quadrature is asymptotically inferior for $(d_i-1)(s_i+1/2)>s_id_i$, i.e. for $d_i>2s_i+1$.
		\\
		The error rate of the SG quadrature depends highly on the underlying one-dimensional rule which constitutes the basis for the sparse grid construction. Its order of accuracy $p_i$ determines with the formula $s_i=\min\{r_i,p_i\}$ whether the SG rule can make full use of the provided regularity $r_i$ of the integrand. For classical one-dimensional Gaussian rules like the Gauss-Hermite or the Gauss-Laguerre rule, we have $p_i=r_i$ for any $r_i\geq 2$, i.e. Gaussian rules always achieve the maximally possible rate \cite{GerstnerGriebel1998}. On the other hand, Newton-Cotes formulas have fixed $p_i$ depending on their construction, e.g. $p_i=2$ for the trapezoid rule. For the Clenshaw-Curtis rules the order of accuracy $p_i=r_i$ for any $r_i\geq 2$ similar to the Gaussian rules could be observed and was proven for $r_i$-times continuously differentiable functions \cite{brasspetras2011}.
		\\
		We are interested in achieving optimal convergence rates, hence in this paper we only consider Gaussian and Clenshaw-Curtis rules as basis for the SG quadrature. Therefore, we have $s_i=r_i$ in the convergence rate of SG quadrature as presented in Table \ref{tab:ConvRatesQi}. This holds similar for Frolov cubature and optimal weights quadrature as they both also achieve the maximally possible main rate $s_i=r_i$.
		\\
		Most of these quadrature rules are designed for integration on the unit cube $\Omega_i=[0,1]^{d_i}$ and can be \textcolor{commentcolor}{linearly transformed to any bounded hypercube} without ramifications for the convergence behavior. However, we also want to examine integrals on unbounded domains, e.g. on $\Omega_i=\R^{d_i}$ or $\Omega_i=[0,\infty)^{d_i}$. In general, we consider two approaches to deal with these integrals: First, if the integrand includes a weight function $\omega_i(x_i)$ on $\Omega_i$, some of the previously mentioned quadrature rules offer adaptations for certain combinations of an integration domain $\Omega_i$ and a weight function $\omega_i$. These adaptations include MC sampling with density function $\omega_i$ (if $\omega_i$ is a p.d.f.\ on $\Omega_i$) and SG quadrature based on certain Gaussian rules. For example, the Gauss-Hermite rule corresponds to the weight function $\omega_i(x_i)=e^{-x_i^2}$ and $\Omega_i=\R$ whereas the Gauss-Laguerre rule corresponds to the weight function $\omega_i(x_i)=e^{-x_i}$ and $\Omega_i=[0,\infty)$. Hence, we can construct SG quadratures from these one-dimensional rules which are defined on $\Omega_i=\R^{d_i}$ or $\Omega_i=[0,\infty)^{d_i}$ with the respective multi-dimensional equivalents of the associated weight functions. The second approach to deal with integrals on unbounded domains is to transform these domains to the unit cube $[0,1]^{d_i}$. Yet depending on the used transformation and the behavior of the original integrand for $x_i\ra\infty$, $x_i\in\Omega_i$, this approach may produce boundary singularities for the transformed integrand. These boundary singularities can drastically reduce the regularity of the integrand and hence prevent us from achieving high convergence rates with higher-order quadrature rules.
		\\
		A recently developed extension of SG quadrature based on Gaussian rules is designed to handle this issue. Oettershagen and Griebel \cite{GriebelOettershagen2014} propose to use SG quadratures which are based on a generalized Gauss formula. This formula is generated similarly to conventional Gaussian formulas with the only difference that polynomials in $\psi(x_i)$ are used instead of polynomials in $x_i$ to compute nodes and weights. In \cite{GriebelOettershagen2014} it was proposed to use $\psi(x_i)=-\log(1-x_i)$, $\psi(x_i)=\arcsinh(2\arctanh(x_i)/\pi)$ or $\psi(x_i)=\erf^{-1}(x_i)$, depending on whether the transformation induces singularities at one or both boundaries of $[0,1]$. This way, $r_i$-times differentiable functions with boundary singularities are included in the space of functions and a main rate of $O(\Ni^{-r_i})$ is achieved for such Gauss-based quadratures. This property is preserved for multidimensional integrands and SG quadrature.
		\\	
		\begin{sidewaystable}
			\footnotesize
			\begin{adjustwidth}{-4.2in}{-4in}  
				\begin{center}
					\renewcommand{\arraystretch}{1.45}
					\begin{tabular}{c|c|c|Sc}
						\backslashbox[60mm]{$\QuadII_{l_2}=$}{$\QuadI_{l_1}=$} & MC & QMC & SG/Frolov/optimal weights \\
						\hline
						MC & $O((\NFG)^{-1/4})$ & $O((\NFG)^{-1/3}\log(\NFG)^{d_1})$ & $O((\NFG)^{-\frac{s_1}{2s_1+1}}\log(\NFG)^{\tilde{t}^\infty_1})$\\
						QMC & $O((\NFG)^{-1/3}\log(\NFG)^{d_2})$ & $O((\NFG)^{-1/2}\log(\NFG)^{d_1+d_2})$ & $O((\NFG)^{-\frac{s_1}{s_1+1}}\log(\NFG)^{\tilde{t}^\infty_1+d_2})$ \\ 
						SG/Frolov/optimal weights & $O((\NFG)^{-\frac{s_2}{2s_2+1}}\log(\NFG)^{\tilde{t}^\infty_2})$ & $O((\NFG)^{-\frac{s_2}{s_2+1}}\log(\NFG)^{d_1+\tilde{t}^\infty_2})$ & $O((\NFG)^{-\frac{s_1s_2}{s_1+s_2}}\log(\NFG)^{\tilde{t}^\infty_1+\tilde{t}^\infty_2})$
					\end{tabular}
				\end{center}
			\end{adjustwidth}
			\hspace{.2\textwidth}
			\caption[caption]{Theoretical convergence rates for FTP quadrature with optimal $\sigma^*$ and $\alpha=1$, where $s_i=\min\{r_i,p_i\}$ for the order of accuracy $p_i$ of the respective quadrature.}
			\label{tab:ConvRatesFG}
			\begin{adjustwidth}{-4in}{-4in}  
				\begin{center}
					\renewcommand{\arraystretch}{1.3}
					\begin{tabular}{c|c|Sc|Sc}
						\backslashbox[55mm]{$\QuadII_{l_2}=$}{$\QuadI_{l_1}=$} & MC & QMC & SG/Frolov/optimal weights \\
						\hline
						MC & - & $O\left((\NSG)^{-1/2}\log(\NSG)^{d_1+1}\right)$ & $O\left((\NSG)^{-1/2}\log(\NSG)^{\tilde{t}^1_1+1}\right)$\\
						QMC & $O\left((\NSG)^{-1/2}\log(\NSG)^{d_2+1}\right)$ & - & $O\left((\NSG)^{-1}\log(\NSG)^{\tilde{t}^1_1+d_2+1}\right)$ \\ 
						SG/Frolov/optimal weights & $O\left((\NSG)^{-1/2}\log(\NSG)^{\tilde{t}^1_2+1}\right)$ & $O\left((\NSG)^{-1}\log(\NSG)^{d_1+\tilde{t}^1_2+1}\right)$ & \makecell{$O\Big((\NSG)^{-\min(s_1,s_2)}\times$\\$\log(\NSG)^{\tilde{t}^1_1+\tilde{t}^1_2+1}\Big)$}
					\end{tabular}
				\end{center}
			\end{adjustwidth}
			\hspace{.2\textwidth}
			\caption[caption]{Theoretical convergence rates for STP quadrature with optimal $\sigma^*\neq1$, $\alpha=1$ and $s_1/\sigma^*=s_2\sigma^*$, where $s_i=\min\{r_i,p_i\}$ for the order of accuracy $p_i$ of the respective quadrature. (The two empty cells are due to $s_1/\sigma^*=s_2\sigma^*$ implying $\sigma^*=1$ in these cases, contrary to the assumption for this table.)}
			\label{tab:ConvRatesSGneq1}
			\begin{adjustwidth}{-4in}{-4in}  
				\begin{center}
					\renewcommand{\arraystretch}{1.3}
					\begin{tabular}{c|Sc|Sc|Sc}
						\backslashbox[55mm]{$\QuadII_{l_2}=$}{$\QuadI_{l_1}=$} & MC & QMC & SG/Frolov/optimal weights \\
						\hline
						MC & $O\left((\NSG)^{-1/2}\log(\NSG)^{3/2}\right)$ & $O\left((\NSG)^{-1/2}\log(\NSG)^{d_1+1/2}\right)$ & $O\left((\NSG)^{-1/2}\log(\NSG)^{\tilde{t}^1_1+1/2}\right)$\\
						QMC & $O\left((\NSG)^{-1/2}\log(\NSG)^{d_2+1/2}\right)$ & $O\left((\NSG)^{-1}\log(\NSG)^{d_1+d_2+2}\right)$ & $O\left((\NSG)^{-1}\log(\NSG)^{\tilde{t}^1_1+d_2+1}\right)$ \\ 
						SG/Frolov/optimal weights & $O\left((\NSG)^{-1/2}\log(\NSG)^{\tilde{t}^1_2+1/2}\right)$ & $O\left((\NSG)^{-1}\log(\NSG)^{d_1+\tilde{t}^1_2+1}\right)$ & \makecell{$O\Big((\NSG)^{-\min(s_1,s_2)}\times$\\$\log(\NSG)^{\tilde{t}^1_1+\tilde{t}^1_2+\min(s_1,s_2)}\Big)$}
					\end{tabular}
				\end{center}
			\end{adjustwidth}
			\hspace{.2\textwidth}
			\caption[caption]{Theoretical convergence rates for STP quadrature where $\sigma^*=1$ is optimal (so an additional summand $\gamma_1^*=\min(s_1,s_2)$ is introduced in the $\log$-exponent), $\alpha=1$, $s_1\neq s_2$ (for $s_1= s_2$  the $\log$-exponent in the lower right cell is increased by 1) and $s_i=\min\{r_i,p_i\}$ for the order of accuracy $p_i$ of the respective quadrature.}
			\label{tab:ConvRatesSG=1}
		\end{sidewaystable}
		Based on the results of Theorems \ref{TPconvergence} and \ref{STPconvergence}, the Tables \ref{tab:ConvRatesFG}, \ref{tab:ConvRatesSGneq1} and \ref{tab:ConvRatesSG=1} now give an overview of the expected convergence rates for various combinations of $\QuadI_{l_1}$ and $\QuadII_{l_2}$. In general, we see that the overall convergence rate is always bounded by the smaller of the two rates of $\QuadI_{l_1}$ and $\QuadII_{l_2}$. In cases where the rate for one of the two integrals is bounded, e.g. by the regularity of the integrand or because quadrature nodes can only be obtained by observations, this rate determines the maximal combined convergence rate. We see that the FTP approach achieves this rate only by using a higher order rule for the other integral. Yet the Tables \ref{tab:ConvRatesSGneq1} and \ref{tab:ConvRatesSG=1} show that the optimal (main) rate can be achieved with any complimentary rule of at least equal order via STP quadrature with optimal $\sigma^*$. This performance is especially impressive if both formulas have the same order: For the same number of nodes the order of the main rate of STP is \textit{doubled} compared to FTP. This is also exactly the behavior which can be observed for traditional sparse grid quadrature, justifying the treatment of STP as a generalized version of SG quadrature.
	
		\subsection{Results for a two-dimensional test function}\label{NR test}
		We start our numerical experiments with a synthetic test function in order to demonstrate the general applicability of the suggested methods. To this end, let
		\begin{align}\label{DummyDef}
			\innerintegrand_\text{test}(z,u|\th):=u^{z+\th-1}e^{-u}
		\end{align}
		and $\OM=[0,1],\ \OMM=[0,\infty)$. Then 
		\begin{align*}
			\OQ(\th)
			&=\int_{0}^{1}\log\left(\int_{0}^{\infty}\innerintegrand_\text{test}(z,u|\th) du\right)dz=\int_{0}^{1}\log\left(\Gamma(z+\th)\right)dz\\
			&=-\log(\Gamma(\th))-\th+(\th-1)\log(\th)+\log(\Gamma(\th+1))+\frac{1}{2}\log(2\pi)\,,
		\end{align*}
		where $\Gamma$ denotes the Gamma function. For our computations we take $\th=4$.
		\\
		Although the general setup is intended for multidimensional integration, this simple example with one-dimensional domains $\OM$ and $\OMM$ already illustrates the improvements resulting from the STP approach. In particular, $\innerintegrand_\text{test}$ is smooth, so $\innerintegrand_\text{test}\in H^{r_1}(\OM)\times H^{r_2}(\OMM)$ for any $r_1,r_2>0$. This implies that every presented quadrature method achieves its maximal order of convergence in both, the inner and the outer integral, which allows for the direct comparison of the theoretical and the numerically observed rates.
		\\
		The choice $F(z,\th,t)=\log(t)$ resembles the Maximum Likelihood setup. Written in the general framework (\ref{IntGMM}) the estimator based on the log-likelihood is constructed with some function $\innerintegrand$ of an unobservable variable which integrates to the choice probability $P(z|\th)=\int_{\OMM}\innerintegrand(z,u|\th)du$. Theorems \ref{TPconvergence} and \ref{STPconvergence} required $F$ to be H\"older continuous in $t$ for all $z\in\OM$ and $\th\in\Theta$. The logarithm is in fact \textit{Lipschitz} continuous, i.e. H\"older continuous with constant $\alpha=1$, but only for $P(z|\th)\geq\delta>0$ for some small constant $\delta$. In an econometric context, it makes sense to assume that $P(z|\th)>0$ but the additional bound $\delta>0$ is less easy to justify. As ML-estimation also requires the integration over $z$, we will assume that such a bound can be prescribed by the choice of the search region $\th$. Additionally, the logarithm function is smooth in $(0,\infty)$, hence $F(z,\th,\cdot)\in H^{r_1}(\R)$ for any $r_1>0$, $z\in\OM$ and $\th\in\Theta$ and $F$ is obviously also smooth in $z$ and $\th$.
		\begin{figure}[!t]
			\centering
			\includegraphics[width=.9\textwidth, height=.8\columnwidth]{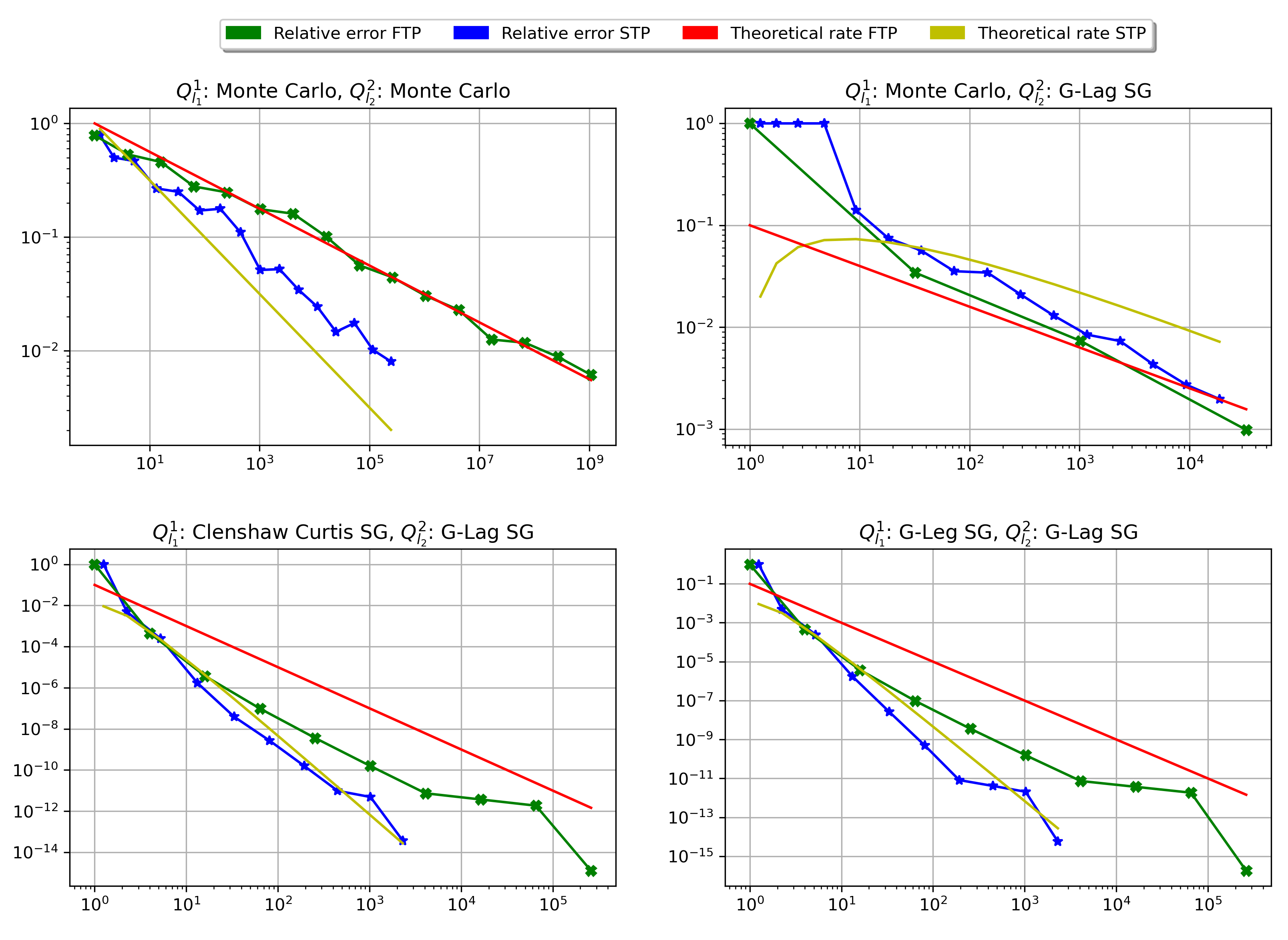}
			\caption[caption]{Full (FTP) and sparse (STP) tensor product quadrature for $F(z,t)=\log(t)$ and $\varphi_\text{test}$ with $\th=4$.
				%\\\hspace{\textheight}
				General legend for all figures: "MC" = Monte Carlo quadrature, ``Sobol" = Quasi Monte Carlo quadrature with Sobol points, ``Halton" = Quasi Monte Carlo quadrature with Halton points, ``Frolov" = Frolov cubature, ``Clenshaw-Curtis SG" = Sparse grid quadrature for the Clenshaw-Curtis rule, ``G-Leg SG" = Sparse grid quadrature for the Gauss-Legendre rule, ``G-Lag SG" = Sparse grid quadrature for the Gauss-Laguerre rule, ``G-Herm SG" = Sparse grid quadrature for the Gauss-Hermite rule, ``Gen-Gauss SG" = Sparse grid quadrature with a generalized Gauss rule with polynomials in $-\log(1-x_i)$, i.e. based on the Gauss-Laguerre rule, ``Optimal weights" = optimal weights MC quadrature with kernel regularity $p_1=5$.}
			\label{fig:MLE-dummy1}
		\end{figure} 
		\\
		Figure \ref{fig:MLE-dummy1} presents four combinations of quadrature formulas for $\Qi$, $i=1,2$. We use for $\QuadI_{l_1}$ MC quadrature, the Clenshaw-Curtis and the Gauss-Legendre rule and for $\QuadII_{l_2}$ MC quadrature and the Gauss-Laguerre rule. The Clenshaw-Curtis and the Gauss-Legendre rule are linearly transformed from the interval $[-1,1]$ to the integration domain $[0,1]$, hence their approximation behavior remains unaffected. The Gauss-Laguerre rule is particularly suited for the inner integral since we actually integrate $u^{z+\th-1}$ with the weight function $\omega_2(u)=e^{-u}$. MC quadrature on $[0,\infty]$ can simply be applied via sampling from the exponential distribution. This way there is no transformation of the inner integral necessary.
		\\
		All plots on the left hand side display the expected better rate of STP versus FTP quadrature. The generally higher convergence of STP quadrature demonstrates the shift from $\gFGopt$ to $\gSGopt$ best. A similar result is obtained for the combination of the Gauss-Legendre rule with the Gauss-Laguerre rule, since both also achieve high rates on their own. In contrast, the combination of Monte Carlo integration with the Gauss-Laguerre rule supports the claim that the slow convergence of the MC rule can barely be ameliorated by the use of higher order formulas for the other integral.
			
		\subsection{Mixed Logit Likelihood} \label{NumExp:MixL}
		
		We continue with the integrals arising from Discrete Choice models encountered in section \ref{DCM}, which we now write in terms of the functional $\OG$ defined in (\ref{DoubleIntegral}), i.e. we specify the domains $\Qi$ for $i=1,2$, the functions $\innerintegrand$ and $F$ and the measures $\mu$ and $\nu$.
		\\
		For both estimators, the objective function $R_N$ denotes the approximation of an integral over the full data space from the real world $\OM$ via
		\begin{align*}
			G_{N_{1,l_1}}(\th)\approx G_\infty:=\int_{\OM}m(z|\th)d\nu(z)=\int_{\OM}F\left(z,\th,\int_{\OMM}\varphi(z,u|\th)d\mu(u)\right)d\nu(z)\,.
		\end{align*}
		While it might be easy to quantify the range of the data ($\OM\subset\R^{d_1}$) it is much harder to determine their distribution in $\OM$. In particular, we cannot choose the quadrature nodes $z_{n_1}^{l_1}$, $n_1=1,...,N_{1,l_1}$, for the approximation deterministically. Hence, the sampling of data points is inherently random and limits the choice of $\QuadI_{l_1}$ for the outer integral in (\ref{DoubleIntegral}) to quadrature methods which are based on random nodes. Therefore, we only consider Monte Carlo and optimal weights cubature for $\QuadI_{l_1}$ and combine them with low- (MC), medium- (Sobol) and high-order (SG or Frolov) rules for $\QuadII_{l_2}$. 
		\\
		We shortly want to point out that the index $n$ for $n=1,...,N$ as it was used in Section \ref{DCM} now changes to $n_1$ with $n_1=1,...,N_{1,l_1}$ in the setting for the approximation, since the summation over the observations is now considered an approximation $Q^1_{l_1}$ of the integral over the data space. Hence, we have now $z_{n_1}^{l_1}$ where we had $z_n$ in Section \ref{DCM}.
		\\
		We can now specify the Mixed Logit model by
		\begin{align*}
			\innerintegrand_\text{MixL}(z,u|\th)=\frac{e^{(zu)_i}}{\sum_{j=1}^{J}e^{(zu)_j}}
		\end{align*}
		and $F(z,\th,t)=\log(t)$. We let $\OM=[0,1]^{d_1}$, where $d_1=J\cdot q$, let $\nu$ be the uniform distribution and set $J=3$ and $q=4$. We consider a multivariate Gaussian distribution as mixing distribution for $u$, hence $\OMM=\R^4$, and $\mu$ resembles the corresponding c.d.f., so that $\Int_1$ denotes a 12-dimensional and $\Int_2$ denotes a 4-dimensional integral, respectively. Finally, we set the parameter vector $\th=(u_0,\Sigma)$ by letting the mean of the mixing distribution be $u_0=(0,0,0,0)\in\R^4$ and letting its covariance matrix $\Sigma=\Sigma_\rho$ be parameterized by $\rho=0.1$ where $\Sigma_{ii}=1$ and $\Sigma_{ij}=\rho$ \textcolor{commentcolor}{ for $i\neq j$}.
		\\
		Similar to the synthetic case from Subsection \ref{NR test}, $\innerintegrand_\text{MixL}$ is smooth. Thus, in theory any order of convergence could be obtained asymptotically. However, the possibly problematic issue in terms of $F$ for $P(z|\th)$ being very close to 0 remains.
		\\
		As the inner integral is defined on $\R^{d_2}$ with a multivariate normal density function, SG quadrature based on the Gauss-Hermite rule and MC quadrature with sampling from a multivariate normal distribution can be applied directly for the inner integral. Furthermore, using a proper \textit{tangens}-transformation, the integral can be transformed to $(0,1)^{d_2}$, so we can also use QMC rules like Sobol quadrature and Frolov cubature on the inner integral. The normal distribution has very flat tails and therefore cancels out any boundary singularities that would have been introduced by the tangens-transformation otherwise.		
		\begin{figure}[!t]
			\centering
			\includegraphics[width=.9\textwidth, height=1.2\columnwidth]{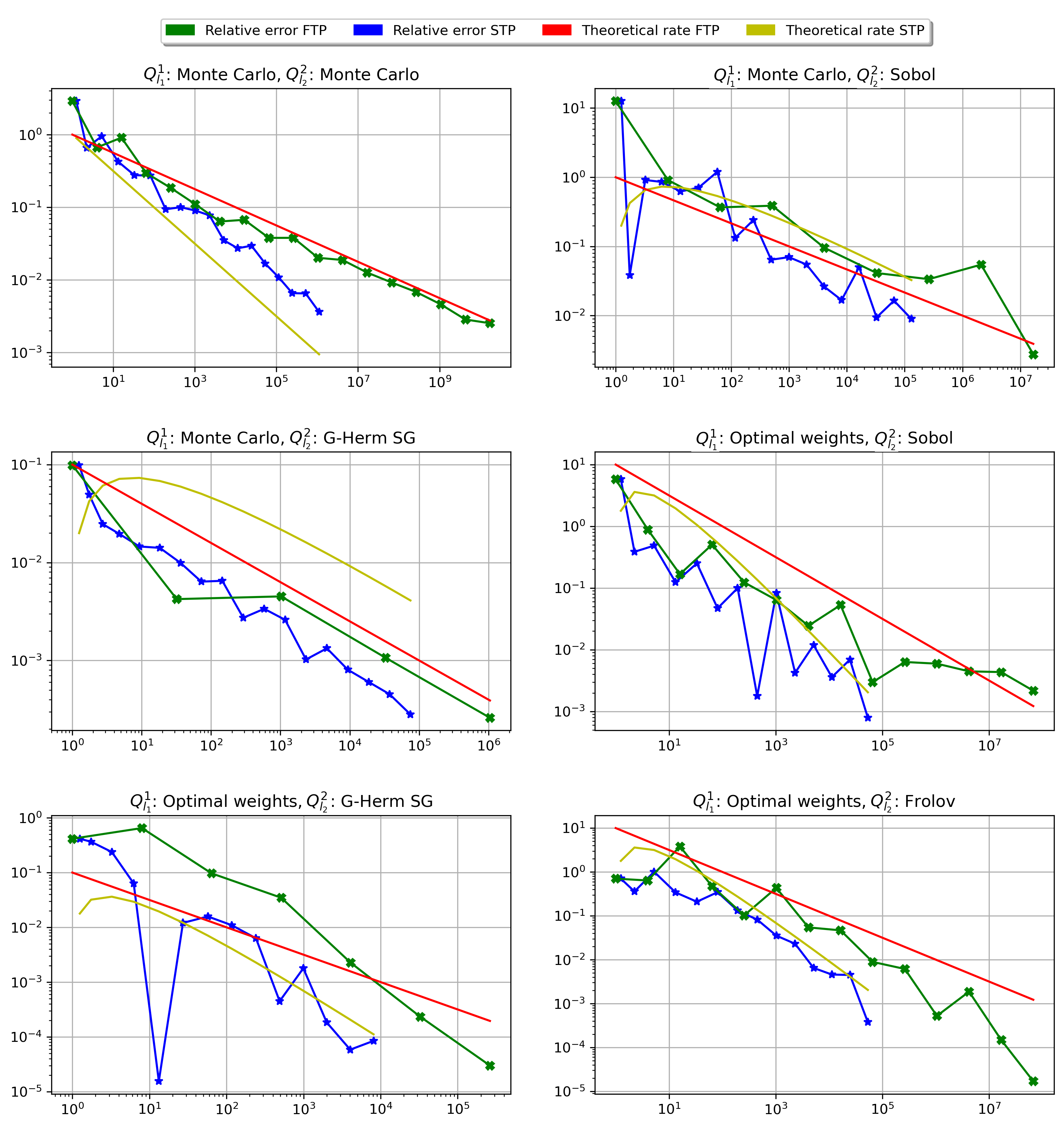}
			\caption{Full (FTP) and sparse (STP) tensor product quadrature for the Mixed Logit model estimated with Maximum Likelihood. (Legend as in Figure \ref{fig:MLE-dummy1})}
			\label{fig:MLE-Mixl}
		\end{figure}
		\\
		Figure \ref{fig:MLE-Mixl} now supports the claims made in Section \ref{MLE-theory}: If MC integration is used for $\QuadI_{l_1}$, then MC integration for $\QuadII_{l_2}$ also gives the highest improvement of the usage of STP compared with FTP. Moreover, a sparse grid rule for $\QuadII_{l_2}$ improves FTP and STP simultaneously. Similar results were obtained for optimal weights cubature where the difference between STP and FTP quadrature can be observed clearly for all combinations.
		\subsection{Multinomial Probit with a GMM estimator}
		
		Based on the seminal paper by McFadden \cite{McFadden1989}, we now investigate the double integral arising from the estimation of the Multinomial Probit model with a GMM-estimator. The associated moment function is defined in (\ref{GMM-momentfunction}).
		\\
		According to the definition (\ref{ProbitChoiceProb}) of $P_{n_1,j}(\th)$, the choice probability is given as the c.d.f. of a multivariate Gaussian distribution, so the derivatives $\nabla_\th P_{n_1,j}(\th)$ exist and can be derived via the corresponding p.d.f.. Furthermore $P_{n_1,j}(\th)$ needs to be computed only for the case $y_{n_1,j}=1$, so the approximation problem for this estimator boils down to the computation of one Multinomial Probit integral for each node/data point $z_{n_1}^{l_1}$.
		\\
		Yet the Multinomial Probit choice probability defined by (\ref{ProbitChoiceProb}) cannot be well approximated directly by the given quadrature rules since the kink introduced by the characteristic function in the integrand reduces the regularity of the integrand drastically. Therefore, higher-order quadrature rule cannot improve upon MC quadrature. But the Genz-algorithm \cite{genz2000}, which is equivalent to the GHK-simulator \cite{hajivassiliou1994classical}, \cite{keane1994solution}, transforms the integral to the unit cube
		\begin{align} \label{ProbitChoiceProbGenz}
			\Int(\mathbf{w})=\int_{(0,1)^{d_2}}\prod_{i=1}^{d_2+1}\hat{w}_i(u_1,...,u_{i-1})du
		\end{align}
		where $d_2=J-2$ for the number of choices $J$ and 
		\begin{align} \label{ProbitUtilities}
			\mathbf{w}=(\tilde{W}_{ki})_{i=1,i\neq k}^{J}=((z\th)_k-(z\th)_i)_{i=1,i\neq k}^{J}
		\end{align}
		for a fixed choice $k\in\{1,...,J\}$. The $\hat{w}_i$ are recursively defined by
		\begin{align*}
			\hat{w}_i(u_1,...,u_{i-1})=\Phi\left(C_{ii}^{-1}\cdot\left(\mathbf{w}_i-\sum_{j=1}^{i-1}C_{ij}\Phi^{-1}(u_j\hat{w}_j(u_1,...,u_{j-1}))\right)\right)
		\end{align*}
		for $i=1,...,J-1$. Here, $\Phi$ is again the c.d.f.~of the standard univariate Gaussian and $C$ is a factor from the Cholesky decomposition of $\Sigma$, i.e. $\Sigma=C^TC$. The inverse c.d.f. $\Phi^{-1}$ induces a boundary singularity for the integrand in (\ref{ProbitChoiceProbGenz}) but it is still analytic away from the boundary. 
		\\
		In the setting of tensor product integration, the definition (\ref{ProbitChoiceProbGenz}) now yields the inner integrand
		\begin{align*}
			\innerintegrand_\text{MNP}(z,u|\th)=\prod_{i=1}^{d_2+1}\hat{w}_i(u_1,...,u_{i-1})
		\end{align*}
		for the Multinomial Probit model. As intermediate function, we obtain
		\begin{align*}
			F(z_{n_1}^{l_1},\th,t)=\frac{1}{t}\nabla_\th P_{n_1,j(n_1)}(\th)-1
		\end{align*}
		based on (\ref{GMM-momentfunction}) and the fact that $y_{n_1,j}\neq0$ only for one alternative $j$, which we denote by $j(n_1)$. 
		\\
		We let again $\OM=[0,1]^{d_1}$ with $d_1=J\cdot q$, we let $\nu$ be the uniform distribution and set $J=5$ and $q=3$, so $\Int_1$ denotes a 15-dimensional integral. The transformed integral (\ref{ProbitChoiceProbGenz}) is defined on the domain $\OMM=(0,1)^{d_2}$ where $d_2=J-2$ and $\mu$ is the corresponding Lebesgue measure, so the inner integral is 3-dimensional. We set $\th=(1,1,1)\in\R^3$ and use the covariance matrix $\Sigma=\Sigma_{0.1}$ similar to the previous subsection. Again, $\innerintegrand_\text{MNP}$ is smooth, so any quadrature formula should achieve its best rate. $F$ is Lipschitz if $t$ is bounded away from 0, thus the conditions of Theorems \ref{TPconvergence} and \ref{STPconvergence} are met.
		\begin{figure}[!t]
			\centering
			\includegraphics[width=.9\textwidth, height=1.2\columnwidth]{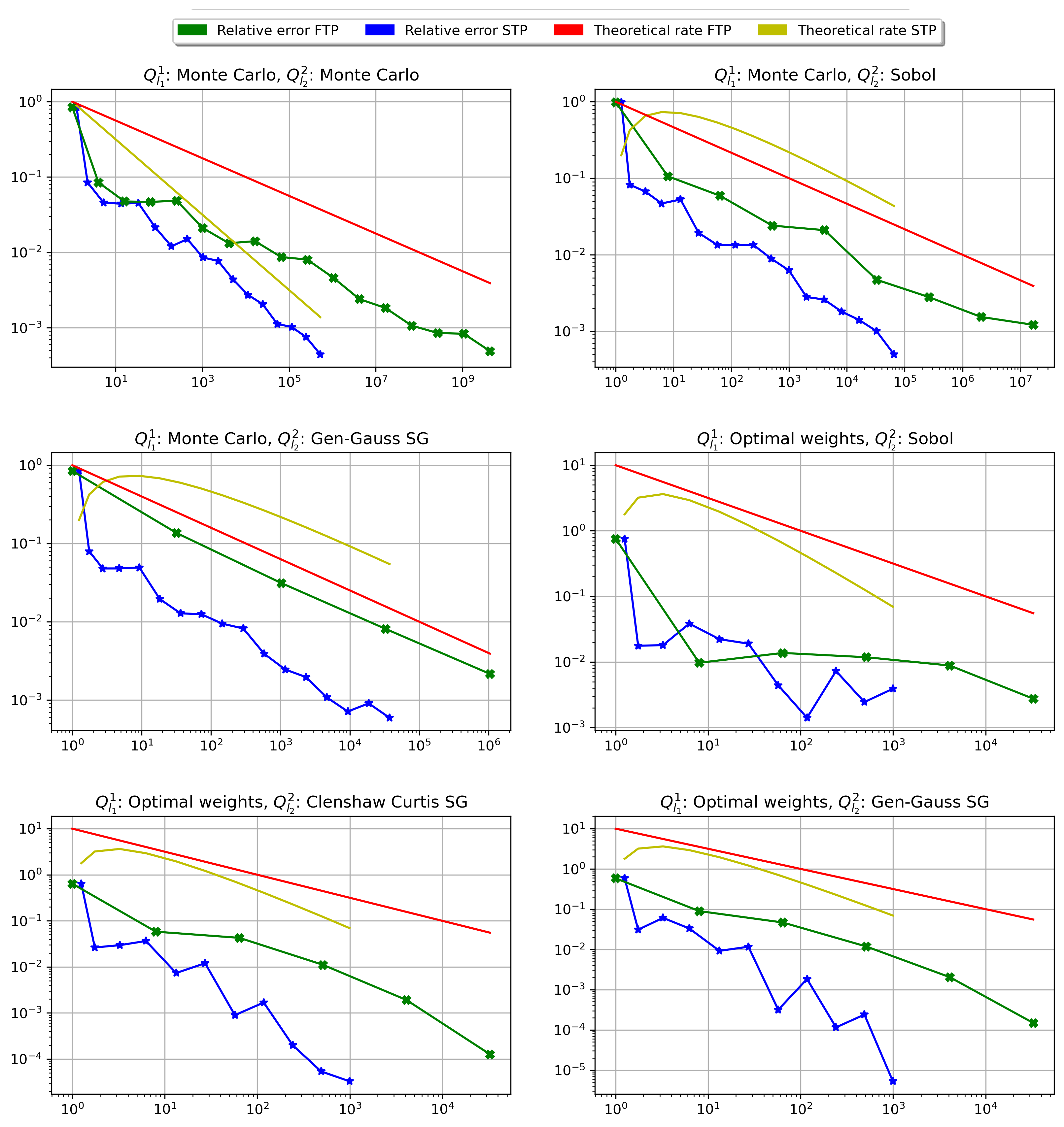}
			\caption{Full (FTP) and sparse (STP) tensor product quadrature for the Multinomial Probit model estimated with GMM. (Legend as in Figure \ref{fig:MLE-dummy1})}
			\label{fig:MLE-MNP}
		\end{figure}
		\\
		In contrast to Subsection \ref{NumExp:MixL}, in this subsection both integration domains are bounded, hence all of the proposed quadrature rules can be applied directly. As before, the Clenshaw-Curtis rule is transformed linearly to $[-1,1]$. However, the Genz transformation still introduces a boundary singularity reducing the regularity of the inner integrand. Thus, we compare not only standard quadrature rules but also apply SG quadrature based on a generalized Gaussian rule with $\psi(u)=-\log(1-u)$ according to the definition above and in \cite{GriebelOettershagen2014}.
		\\
		We see in Figure \ref{fig:MLE-MNP} that STP clearly outperforms FTP for all combinations of MC or optimal weight cubature with low and high order formulas for $\QuadII_{l_2}$. In particular, for Monte Carlo integration for $\QuadI_{l_1}$, STP and FTP follow the expected rates closely, similar to the above case of Mixed Logit/Maximum likelihood. Only the combination of optimal weights with Frolov cubature fails due to the bad performance of the latter which is caused by the non-zero \textcolor{commentcolor}{boundary trace} of the integrand.
		\begin{figure}[!t]
			\centering
			\includegraphics[width=.9\linewidth, height=.8\columnwidth]{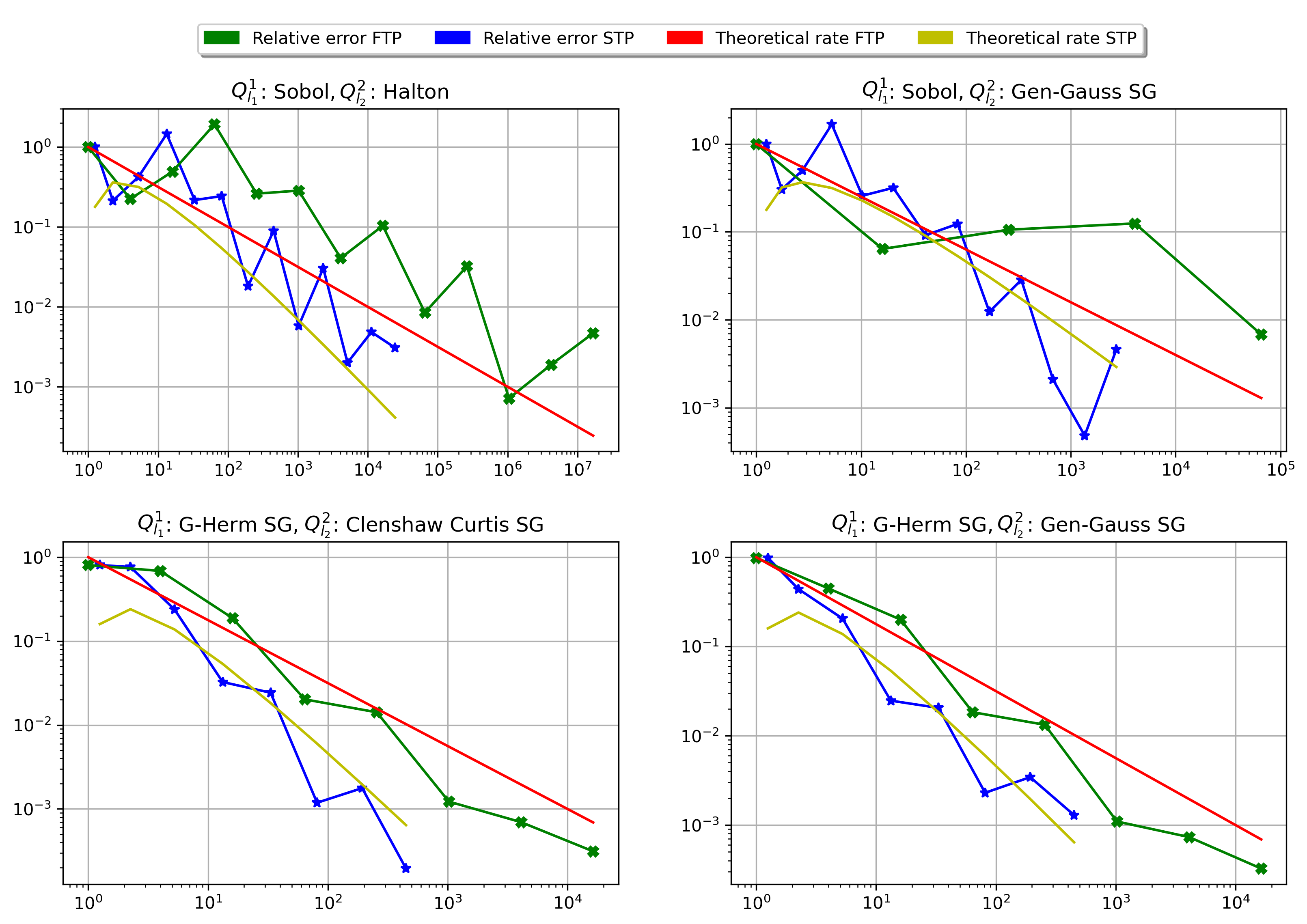}
			\caption{Full (FTP) and sparse (STP) tensor product quadrature for the Mixed Multinomial Probit model. (Legend as in Figure \ref{fig:MLE-dummy1})}
			\label{fig:MLE-MixMNP}
		\end{figure}
		
		\subsection{Mixed Probit}
		
		As final example, we recall the Mixed Probit model from Section \ref{DCM}: There we noticed that, although the multivariate Probit model already allows for correlation between choices, a mixture distribution might be superior in some cases. However, a Mixed Probit model is computationally even more challenging, since it involves not only the approximation of a multivariate Gaussian distribution but also an approximation of the integral over the parameter mixture. In particular, the multivariate Gaussian has to be calculated at every quadrature node for the \textcolor{commentcolor}{integral over the mixing distribution}. 
		\\
		Hence, we have again two nested integrals for which we can compare FTP and STP quadrature. The inner integrand remains $\innerintegrand_\text{MNP}$ (since the Multinomial Probit choice probability is the basis for the Mixed Probit model) with covariance matrix $\Sigma=\Sigma_{0.2}$ and similar specification of $\OMM$ and $\mu$. But in the computation of the utility $V$, the roles of integration and fixed variable are interchanged: We now integrate over $\th$ and fix a set of observed variables $z$, thus, in the established notation, we have $``z"=``\th"$ and $z$ is a value in the new parameter vector $\th$. This also changes the dimensionality of $\OM$ from $q\cdot J$ to $q$. We draw $q\cdot J$ values randomly from a uniform distribution to assemble the new parameter $z\in\R^{q\times J}$ and set $q=4$ and $J=5$ obtaining 3-dimensional inner and 4-dimensional outer integrals, respectively.
		\\
		Furthermore, the intermediate function $F$ becomes
		\begin{align*}
			F(z,\th,t)=t\cdot h(z|\th_0,\Xi)
		\end{align*}
		for a mixing distribution $h$ with mean $\th_0$ and covariance matrix $\Xi$. As for the Mixed Logit model, we let $h$ be a multivariate Gaussian distribution, set $\th_0=(0.2,0.2,0.2,0.2)\in\R^4$ and $\Xi=\Sigma_{0.1}$, and have the complete parameter vector $(z,\Sigma,\th_0,\Xi)$.
		\\
		For the inner integral, the same quadrature rules are available as in the previous section (with the same linear transformation for the Clenshaw-Curtis rule). Yet we are no longer restricted to MC and optimal weights cubature for the outer integral as the integration is now \textcolor{commentcolor}{also induced by the model and not by the simulation}. Therefore, we can test STP quadrature for higher-order rules. As the outer integral is defined on $\OM=\R^{d_1}$ with a multivariate normal density function, the same restrictions apply as for the inner integral in Section \ref{NumExp:MixL}, i.e. we can simply use SG quadrature based on the Gauss-Hermite rule or any of the other rules if we first apply a tangens-transformation for the outer integral.
		\\
		We display the resulting improvements in Figure \ref{fig:MLE-MixMNP}. Once more, the results underscore the predictions we made in Tables \ref{tab:ConvRatesFG},\ref{tab:ConvRatesSGneq1} and \ref{tab:ConvRatesSG=1} as STP clearly outperforms FTP quadrature. Furthermore, we see how the benefits are more visible if same-order rules are used for $\QuadI_{l_1}$ and $\QuadII_{l_2}$ and we can observe how the high order of SG quadrature is sustained.
		
	\section{Concluding remarks} \label{Conclusion}
		In the present paper, we adapted the sparse tensor product (STP) technique for integration problems, proved corresponding theorems on error bounds and derived the optimal balancing factor. In particular, we proposed a H\"older continuity condition for the linking function to preserve error bounds. It turned out that the improvements of STP quadrature compared with the classic balancing approach between inner and outer integral (FTP quadrature) are most significant if the rules used for the inner and outer integral achieve similar convergence rates. Then, the order of the total rate is almost doubled for STP quadrature and almost equals the rate of each separate formula.
		\\
		We then presented two popular models from Discrete Choice modeling which both require the approximation of a multi-dimensional integral. Furthermore, we discovered that M- and GMM-estimators can be considered as Monte Carlo simulations of integrals over the domain of the observable variables, i.e. the ``real world data" space. Together with the integrals posed by the respective models and objective functions they comprise nested integrals separated by an intermediate function.		
		\\
		Introducing STP quadrature to these nested integrals, we combined the approximation of inner and outer integral to obtain significantly improved approximations of Maximum Likelihood- and GMM-estimators, where Mixed Logit and Multinomial Probit models served as examples for models with multidimensional integrals which are not analytically solvable. For both instances, as well as a synthetic test function, the proposed STP quadrature approach was considerably better than the standard approach, achieving up to twice the original order of convergence. Finally, for the Mixed Multinomial Probit model, which directly included a nested integral, STP quadrature was similarly effective.
		\\
		We conclude that econometric estimation and nested integrals arising from econometric models can significantly benefit from using sparse tensor product quadrature. In estimation, which we simulated using MC sampling in the domain $\OM$, this enables us to reach the best possible main rate $O(N^{-1/2})$ for any rule $\QuadII_{l_2}$ and hence to increase the accuracy of an ML- or GMM-estimator for a fixed set of observations. For nested integrals as in the Mixed Probit model, we preserve polynomial (and possibly even exponential) convergence rates and make intractable models numerically feasible.
		\\
		\textcolor{commentcolor}{For applications where {\em both} integrals are induced by the model (and not one by simulation and one by the model), the regularity of the integrands is decisive for the overall achievable convergence rate. It may happen that the advantage of higher-order quadrature rules gets eliminated if the integrands are not sufficiently smooth. Depending on the model, the parameter $\th$ might also affect the convergence behavior of the inner and outer quadratures, e.g. certain choices of $\th$ might introduce singularities or kinks to the integrands or just might move them closer to such irregularities.}
		\\
		\textcolor{commentcolor}{Eventhough the sparse tensor product effect requires little regularity of the integrands $\innerintegrand$ and $F$, convergence rates higher than $O(N^{-1/2})$ can in general not be achieved even for smooth functions, if the outer integral is approximated only by real world observations. Then, as the quadrature points for the outer integral are chosen at \textit{random} in this setting, just the convergence theory for MC sampling is applicable. Altogether this leads to the best possible rate $O(N^{-1/2})$ even for an analytical inner integrand.}
		\\
		It is to be explored in the future if and how estimation in $\OM$ can be enhanced by the weights from optimally weighted MC and the estimation thus can benefit from its higher rate. Altogether this would result in much faster convergence for both, FTP and STP, while still keeping the advantage of STP over FTP of a doubled rate which we demonstrated in this article.
		\bigbreak
		\textbf{Acknowledgments.} Michael Griebel was supported by the \textit{Hausdorff Center for Mathematics} in Bonn, funded by the Deutsche Forschungsgemeinschaft (DFG, German Research Foundation) under Germany's Excellence Strategy - EXC-2047/1 - 390685813 and the Sonderforschungsbereich 1060 \textit{The Mathematics of Emergent Effects} of the Deutsche Forschungsgemeinschaft. 
	%	\newpage
		\bigbreak

		\printbibliography

@Unpublished{GriebelOettershagen2019,
  author         = {M. Griebel and F. Heiss and J. Oettershagen and C. Weiser},
  title          = {Maximum approximated likelihood estimation},
  year           = {(2019)},
  note           = {University of Bonn, INS Preprint No. 1905},
  __markedentry  = {[alexa:1]},
  annote         = {article, inspreprint},
  inspreprintnum = {1905},
}

@InCollection{GriebelOettershagen2014,
  author         = {Griebel, M. and Oettershagen, J.},
  title          = {Dimension-adaptive sparse grid quadrature for integrals with boundary singularities},
  booktitle      = {Sparse {G}rids and {A}pplications},
  year           = {2014},
  volume         = {97},
  series         = {Lecture Notes in Computational Science and Engineering},
  publisher      = {Springer},
  pages          = {109--136},
  __markedentry  = {[alexa:1]},
  annote         = {series,inspreprint},
  file           = {SingularIntegrals.pdf:http\://wissrech.ins.uni-bonn.de/research/pub/oettershagen/SingularIntegrals.pdf:PDF},
  inspreprintnum = {1310},
}

@PhdThesis{Oettershagen2017,
  author        = {Oettershagen, J.},
  title         = {{Construction of Optimal Cubature Algorithms with Applications to Econometrics and Uncertainty Quantification}},
  year          = {2017},
  type          = {PhD Thesis, Institut f\"ur Numerische Simulation, Universit\"{a}t Bonn},
  __markedentry = {[alexa:1]},
  annote        = {INSdiss,thesis},
  file          = {diss_oettershagen.pdf:http\://wissrech.ins.uni-bonn.de/research/pub/oettershagen/diss_oettershagen.pdf:PDF},
}

@Article{BungartzGriebel2004,
  author        = {Bungartz, H.--J. and Griebel, M.},
  title         = {Sparse grids},
  journal       = {Acta Numerica},
  year          = {2004},
  volume        = {13},
  pages         = {147-269},
  __markedentry = {[alexa:1]},
  abstract      = {We present a survey of the fundamentals and the
		  applications of sparse grids, with a focus on the solution
		  of partial differential equations (PDEs). The sparse grid
		  approach, introduced in Zenger (1991), is based on a
		  higher-dimensional multiscale basis, which is derived from
		  a one-dimensional multiscale basis by a tensor product
		  construction. Discretizations on sparse grids involve $O(N
		  (\log N)^{d-1})$ degrees of freedom only, where $d$ denotes
		  the underlying problem's dimensionality and where $N$ is
		  the number of grid points in one coordinate direction at
		  the boundary. The accuracy obtained with piece-wise linear
		  basis functions, for example, is $O(N^{-2} (\log N)^{d-1})$
		  with respect to the $L_2$- and $L_\infty$-norm, if the
		  solution has bounded second mixed derivatives. This way,
		  the curse of dimensionality, i.e., the exponential
		  dependence $O(N^d)$ of conventional approaches, is overcome
		  to some extent. For the energy norm, only $O(N)$ degrees of
		  freedom are needed to give an accuracy of $O(N^{-1})$. This
		  is why sparse grids are especially well-suited for problems
		  of very high dimensionality.
		  
		  The sparse grid approach can be extended to nonsmooth
		  solutions by adaptive refinement methods. Furthermore, it
		  can be generalized from piecewise linear to higher-order
		  polynomials. Also, more sophisticated basis functions like
		  interpolets, prewavelets, or wavelets can be used in a
		  straightforward way.
		  
		  We describe the basis features of sparse grids and report
		  the results of various numerical experiments for the
		  solution of elliptic PDEs as well as for other selected
		  problems such as numerical quadrature and data mining.},
  annote        = {article,1145,amamef,C2,data,ALM},
  file          = {sparsegrids.pdf:http\://wissrech.ins.uni-bonn.de/research/pub/griebel/sparsegrids.pdf:PDF},
}

@Article{GerstnerGriebel1998,
  author        = {Gerstner, T. and Griebel, M.},
  title         = {Numerical integration using sparse grids},
  journal       = {Numerical Algorithms},
  year          = {1998},
  volume        = {18},
  number        = {3},
  pages         = {209--232},
  __markedentry = {[alexa:1]},
  abstract      = { We present new and review existing algorithms for the
		  numerical integration of multivariate functions defined
		  over $d$--dimensional cubes using several variants of the
		  sparse grid method first introduced by Smolyak. In this
		  approach, multivariate quadrature formulas are constructed
		  using combinations of tensor products of suited
		  one--dimensional formulas. The computing cost is almost
		  independent of the dimension of the problem if the function
		  under consideration has bounded mixed derivatives. We
		  suggest the usage of extended Gauss (Patterson) quadrature
		  formulas as the one--dimensional basis of the construction
		  and show their superiority in comparison to previously used
		  sparse grid approaches based on the trapezoidal,
		  Clenshaw--Curtis and Gauss rules in several numerical
		  experiments and applications. For the computation of path
		  integrals further improvements can be obtained by combining
		  generalized Smolyak quadrature with the Brownian bridge
		  construction. },
  annote        = {article},
  file          = {quad.ps.gz:http\://wissrech.ins.uni-bonn.de/research/pub/gerstner/quad.ps.gz:PostScript},
}

@Article{NovakRitter1996,
  author        = {Novak, Erich and Ritter, Klaus},
  title         = {High Dimensional Integration of Smooth Functions over Cubes},
  journal       = {Numerische Mathematik},
  year          = {1996},
  volume        = {75},
  month         = {10},
  pages         = {79-97},
  __markedentry = {[alexa:1]},
}

@Article{giles2015multilevel,
  author        = {Giles, M.},
  title         = {Multilevel {M}onte {C}arlo methods},
  journal       = {Acta Numerica},
  year          = {2015},
  volume        = {24},
  pages         = {259--328},
  __markedentry = {[alexa:1]},
  publisher     = {Cambridge University Press},
}

@Article{KaOeUlUl2018,
  author         = {Kacwin, C. and Oettershagen, J. and Ullrich, M. and Ullrich, T.},
  title          = {Numerical performance of optimized {Frolov} lattices in tensor product reproducing kernel {Sobolev} spaces},
  journal        = {Found. Comput. Math.},
  volume = {21},
  pages = {849-889},
  year           = {2021},
  __markedentry  = {[alexa:1]},
  annote         = {preprint, group},
  file           = {Frolov-numerics_final.pdf:http\://ullrich.ins.uni-bonn.de/Frolov-numerics_final.pdf:PDF},
  inspreprintnum = {1801},
}

@Article{Heiss2010,
  author        = {Heiss, F.},
  title         = {The panel probit model: {A}daptive integration on sparse grids},
  journal       = {Advances in Econometrics},
  year          = {2010},
  volume        = {26},
  pages         = {41-64},
  __markedentry = {[alexa:1]},
  isbn          = {978-0-85724-149-8},
}

@Article{HeissWinschel2008,
  author        = {Heiss, F. and Winschel, V.},
  title         = {Likelihood approximation by numerical integration on sparse grids},
  journal       = {Journal of Econometrics},
  year          = {2008},
  volume        = {144},
  number        = {1},
  pages         = {62-80},
  __markedentry = {[alexa:1]},
}

@Book{Train2009,
  author        = {Train, K.},
  title         = {Discrete {C}hoice {M}ethods {W}ith {S}imulation},
  year          = {2009},
  publisher     = {Cambridge University Press},
  journal       = {Discrete Choice Methods with Simulation, Second Edition},
}

@Article{butler1982computationally,
  author        = {Butler, J. and Moffitt, R.},
  title         = {A computationally efficient quadrature procedure for the one--factor multinomial probit model},
  journal       = {Econometrica},
  year          = {1982},
  volume        = {50},
  number        = {3},
  pages         = {761--764},
  __markedentry = {[alexa:1]},
  publisher     = {JSTOR},
}

@Book{hayashi2000econometrics,
  author    = {Hayashi, F.},
  title     = {Econometrics},
  year      = {2000},
  publisher = {Princeton University Press},
}

@Article{keane1994solution,
  author        = {Keane, M. and Wolpin, K.},
  title         = {The solution and estimation of discrete choice dynamic programming models by simulation and interpolation: {M}onte {C}arlo evidence},
  journal       = {The Review of Economics and Statistics},
  year          = {1994},
  volume        = {76},
  number        = {4},
  pages         = {648--672},
  __markedentry = {[alexa:1]},
  publisher     = {JSTOR},
}

@Article{WinschelKraetzig2010,
  author        = {Winschel, V. and Kraetzig, M.},
  title         = {Solving, Estimating, and Selecting Nonlinear Dynamic Models Without the Curse of Dimensionality},
  journal       = {Econometrica},
  year          = {2010},
  volume        = {78},
  pages         = {803-821},
  __markedentry = {[alexa:1]},
}

@Article{Malin2011,
  author        = {Malin, B. and Kr\"uger, D. and K\"ubler, F.},
  title         = {Solving the multi-country real business cycle model using a {S}molyak-collocation method},
  journal       = {Journal of Economic Dynamics and Control},
  year          = {2011},
  volume        = {35},
  number        = {2},
  pages         = {229 - 239},
  __markedentry = {[alexa:1]},
  abstract      = {We describe a sparse-grid collocation method to compute recursive solutions of dynamic economies with a sizable number of state variables. We show how powerful this method can be in applications by computing the non-linear recursive solution of an international real business cycle model with a substantial number of countries, complete insurance markets and frictions that impede frictionless international capital flows. In this economy, the aggregate state vector includes the distribution of world capital across different countries as well as the exogenous country-specific technology shocks. We use the algorithm to efficiently solve models with up to 10 countries (i.e., up to 20 continuous-valued state variables).},
  keywords      = {Sparse grids, Collocation, International real business cycles},
}

@Article{KruegerKubler2004,
  author        = {Kr\"uger, D. and K\"ubler, F.},
  title         = {Computing {OLG} models with stochastic production},
  journal       = {Journal of Economic Dynamics and Control},
  year          = {2004},
  volume        = {28},
  number        = {7},
  pages         = {1411-1436},
  __markedentry = {[alexa:1]},
}

@Article{Judd2014,
  author        = {Judd, K. and Maliar, L. and Maliar, S. and Valero, R.},
  title         = {Smolyak method for solving dynamic economic models: {L}agrange interpolation, anisotropic grid and adaptive domain},
  journal       = {Journal of Economic Dynamics and Control},
  year          = {2014},
  volume        = {44},
  pages         = {92-123},
  __markedentry = {[alexa:1]},
}

@Article{BrummScheidegger2017,
  author        = {Brumm, J. and Scheidegger, S.},
  title         = {Using adaptive sparse grids to solve high--dimensional dynamic models},
  journal       = {Econometrica},
  year          = {2017},
  volume        = {85},
  number        = {5},
  pages         = {1575-1612},
  __markedentry = {[alexa:1]},
  abstract      = {We present a flexible and scalable method for computing global solutions of high-dimensional stochastic dynamic models. Within a time iteration or value function iteration setup, we interpolate functions using an adaptive sparse grid algorithm. With increasing dimensions, sparse grids grow much more slowly than standard tensor product grids. Moreover, adaptivity adds a second layer of sparsity, as grid points are added only where they are most needed, for instance, in regions with steep gradients or at nondifferentiabilities. To further speed up the solution process, our implementation is fully hybrid parallel, combining distributed and shared memory parallelization paradigms, and thus permits an efficient use of high-performance computing architectures. To demonstrate the broad applicability of our method, we solve two very different types of dynamic models: first, high-dimensional international real business cycle models with capital adjustment costs and irreversible investment; second, multiproduct menu-cost models with temporary sales and economies of scope in price setting.},
  keywords      = {Adaptive sparse grids, high-performance computing, international real business cycles, menu costs, occasionally binding constraints},
}

@Article{Abay2014,
  author        = {Abay, K.},
  title         = {Evaluating simulation-based approaches and multivariate quadrature on sparse grids in estimating multivariate binary {P}robit models},
  journal       = {Economics Letters},
  year          = {2014},
  volume        = {126},
  __markedentry = {[alexa:1]},
}

@InProceedings{Geweke1998MC,
  author        = {Geweke, J.},
  title         = {Efficient simulation from the multivariate normal and student-t distributions subject to linear constraints and the evaluation of constraint probabilities},
  booktitle     = {Computing {S}cience and {S}tatistics: {P}roceedings of the {T}wenty-{T}hird {S}ymposium on the {I}nterface},
  year          = {1998},
  editor        = {E. Keramidas},
  volume        = {23},
  __markedentry = {[alexa:1]},
  journal       = {Comput. Sci. Statist.},
}

@InProceedings{JuddSkrainka2011,
  author        = {Judd, K. and Skrainka, B.},
  title         = {High performance quadrature rules: {H}ow numerical integration affects a popular model of product differentiation},
  booktitle     = {{CeMMAP} working papers},
  year          = {2011},
  __markedentry = {[alexa:1]},
  abstract      = {Efficient, accurate, multi-dimensional, numerical integration has become an important tool for approximating the integrals which arise in modern economic models built on unobserved heterogeneity, incomplete information, and uncertainty. This paper demonstrates that polynomialbased rules out-perform number-theoretic quadrature (Monte Carlo) rules both in terms of efficiency and accuracy. To show the impact a quadrature method can have on results, we examine the performance of these rules in the context of Berry, Levinsohn, and Pakes (1995)'s model of product differentiation, where Monte Carlo methods introduce considerable numerical error and instability into the computations. These problems include inaccurate point estimates, excessively tight standard errors, instability of the inner loop 'contraction' mapping for inverting market shares, and poor convergence of several state of the art solvers when computing point estimates. Both monomial rules and sparse grid methods lack these problems and provide a more accurate, cheaper method for quadrature. Finally, we demonstrate how researchers can easily utilize high quality, high dimensional quadrature rules in their own work.},
  type          = {CeMMAP working papers},
}

@Article{McFadden1989,
  author        = {McFadden, D.},
  title         = {A Method of Simulated Moments for Estimation of Discrete Response Models without Numerical Integration},
  journal       = {Econometrica},
  year          = {1989},
  volume        = {57},
  number        = {5},
  pages         = {995-1026},
  __markedentry = {[alexa:1]},
  abstract      = {This paper proposes a simple modification of a conventional generalized method of moments estimator for a discrete response model, replacing response probabilities that require numerical integration with estimators obtained by Monte Carlo simulation. This method of simulated moments does not require precise estimates of these probabilities, as the law of large numbers operating across observations controls simulation error, and, hence, can use simulations of practical size. The method is useful for models such as high-dimensional multinomial probit, where computation has previously restricted applications. Statistical properties are established using empirical process methods that can handle discontinuities introduced by simulation. Copyright 1989 by The Econometric Society.},
}

@Book{GourierouxMonfort1997,
  author        = {Gourieroux, C. and Monfort, A.},
  title         = {Simulation-based {E}conometric {M}ethods},
  year          = {1997},
  publisher     = {Oxford University Press},
  __markedentry = {[alexa:1]},
  abstract      = {This book introduces a new generation of statistical econometrics. After linear models leading to analytical expressions for estimators, and non-linear models using numerical optimization algorithms, the availability of high- speed computing has enabled econometricians to consider econometric models without simple analytical expressions. The previous difficulties presented by the presence of integrals of large dimensions in the probability density functions or in the moments can be circumvented by a simulation-based approach. After a brief survey of classical parametric and semi-parametric non-linear estimation methods and a description of problems in which criterion functions contain integrals, the authors present a general form of the model where it is possible to simulate the observations. They then move to calibration problems and the simulated analogue of the method of moments, before considering simulated versions of maximum likelihood, pseudo-maximum likelihood, or non-linear least squares. The general principle of indirect inference is presented and is then applied to limited dependent variable models and to financial series.},
}

@Article{genz2000,
  author        = {Genz, A.},
  title         = {Numerical Computation Of Multivariate Normal Probabilities},
  journal       = {Journal of Computational and Graphical Statistics},
  year          = {2000},
  volume        = {1},
  __markedentry = {[alexa:1]},
}

@Article{Bhat2001,
  author        = {Bhat, C.},
  title         = {Quasi--random maximum simulated likelihood estimation of the mixed multinomial {L}ogit model},
  journal       = {Transportation Research Part B: Methodological},
  year          = {2001},
  volume        = {35},
  number        = {7},
  pages         = {677-693},
  __markedentry = {[alexa:1]},
  abstract      = {This paper proposes the use of a quasi-random sequence for the estimation of the mixed multinomial logit model. The mixed multinomial structure is a flexible discrete choice formulation which accommodates general patterns of competitiveness as well as heterogeneity across individuals in sensitivity to exogenous variables. The estimation of this model has been achieved in the past using the pseudo-random maximum simulated likelihood method that evaluates the multi-dimensional integrals in the log-likelihood function by computing the integrand at a sequence of pseudo-random points and taking the average of the resulting integrand values. We suggest and implement an alternative quasi-random maximum simulated likelihood method which uses cleverly crafted non-random but more uniformly distributed sequences in place of the pseudo-random points in the estimation of the mixed logit model. Numerical experiments, in the context of intercity travel mode choice, indicate that the quasi-random method provides considerably better accuracy with much fewer draws and computational time than does the pseudo-random method. This result has the potential to dramatically influence the use of the mixed logit model in practice; specifically, given the flexibility of the mixed logit model, the use of the quasi-random estimation method should facilitate the application of behaviorally rich structures in discrete choice modeling.},
}

@InCollection{NeweyMcFadden1994,
  author    = {Newey, W. and McFadden, D.},
  title     = {Chapter 36 Large sample estimation and hypothesis testing},
  booktitle = {Handbook of Econometrics},
  year      = {1994},
  volume    = {4},
  publisher = {Elsevier},
  pages     = {2111 - 2245},
}

@Article{DickKuoSloan2013,
  author        = {Dick, J. and Kuo, F. and Sloan, I.},
  title         = {High--dimensional integration: {T}he {Q}uasi--{M}onte {C}arlo way},
  journal       = {Acta Numerica},
  year          = {2013},
  volume        = {22},
  pages         = {133-288},
  __markedentry = {[alexa:1]},
  publisher     = {Cambridge University Press},
}

@Article{GriebelHarbrecht2011,
  author         = {Griebel, M. and Harbrecht, H.},
  title          = {On the construction of sparse tensor product spaces},
  journal        = {Mathematics of Computations},
  year           = {2013},
  volume         = {82},
  number         = {282},
  pages          = {975--994},
  __markedentry  = {[alexa:1]},
  annote         = {article},
  file           = {OnTheConstructionOfSparseTensorProductSpaces.pdf:http\://wissrech.ins.uni-bonn.de/research/pub/griebel/OnTheConstructionOfSparseTensorProductSpaces.pdf:PDF},
  inspreprintnum = {1104},
}

@InProceedings{Heinrich2001,
  author        = {Heinrich, S.},
  title         = {Multilevel {M}onte {C}arlo methods},
  booktitle     = {Large-Scale Scientific Computing},
  year          = {2001},
  editor        = {Margenov, S. and Wa{\'{s}}niewski, J. and Yalamov, P.},
  publisher     = {Springer},
  pages         = {58--67},
  __markedentry = {[alexa:1]},
  abstract      = {We study Monte Carlo approximations to high dimensional parameter dependent integrals. We survey the multilevel variance reduction technique introduced bythe author in [4] and present extensions and new developments of it. The tools needed for the convergence analysis of vector-valued Monte Carlo methods are discussed, as well. Applications to stochastic solution of integral equations are given for the case where an approximation of the full solution function or a family of functionals of the solution depending on a parameter of a certain dimension is sought.},
}

@Article{Wasilkowski1995,
  author        = {G. Wasilkowski and H. Wozniakowski},
  title         = {Explicit Cost Bounds of Algorithms for Multivariate Tensor Product Problems},
  journal       = {Journal of Complexity},
  year          = {1995},
  volume        = {11},
  number        = {1},
  pages         = {1--56},
  __markedentry = {[alexa:1]},
  abstract      = {We study multivariate tenser product problems in the worst case and average case settings. They are defined on functions of d variables. For arbitrary d, we provide explicit upper bounds on the costs of algorithms which compute an ϵ-approximation to the solution. The cost bounds are of the form (c(d) + 2)β1(β2 + β3(ln 1/ϵ)/(d − 1))β4(d − 1)(1/ϵ)β5. Here c(d) is the cost of one function evaluation (or one linear functional evaluation), and βi′s do not depend on d; they are determined by the properties of the problem for d = 1. For certain tensor product problems, these cost bounds do not exceed c(d)Kϵ−p for some numbers K and p, both independent of d. However, the exponents p which we obtain are too large. We apply these general estimates to certain integration and approximation problems in the worst and average case settings. We also obtain an upper bound, which is independent of d, for the number, n(ϵ, d), of points for which discrepancy (with unequal weights) is at most ϵ, n(ϵ, d) ≤ 7.26ϵ−2.454, ∀d, ϵ ≤ 1.},
}

@Book{mcculloch2001G-L-M-M,
  author        = {McCulloch, C. and Searle, S. and Neuhaus, J.},
  title         = {Generalized, Linear, and Mixed Models},
  year          = {2001},
  publisher     = {Wiley},
  __markedentry = {[alexa:1]},
}

@InCollection{Maliar2014,
  author        = {Maliar, L. and Maliar, S.},
  title         = {Numerical Methods for Large-Scale Dynamic Economic Models},
  booktitle     = {Handbook of {C}omputational {E}conomics},
  year          = {2014},
  editor        = {K. Schmedders and K. Judd},
  volume        = {3},
  publisher     = {Elsevier},
  chapter       = {7},
  pages         = {325 - 477},
  __markedentry = {[alexa:1]},
  abstract      = {We survey numerical methods that are tractable in dynamic economic models with a finite, large number of continuous state variables. (Examples of such models are new Keynesian models, life-cycle models, heterogeneous-agents models, asset-pricing models, multisector models, multicountry models, and climate change models.) First, we describe the ingredients that help us to reduce the cost of global solution methods. These are efficient nonproduct techniques for interpolating and approximating functions (Smolyak, stochastic simulation, and ε-distinguishable set grids), accurate low-cost monomial integration formulas, derivative-free solvers, and numerically stable regression methods. Second, we discuss endogenous grid and envelope condition methods that reduce the cost and increase accuracy of value function iteration. Third, we show precomputation techniques that construct solution manifolds for some models’ variables outside the main iterative cycle. Fourth, we review techniques that increase the accuracy of perturbation methods: a change of variables and a hybrid of local and global solutions. Finally, we show examples of parallel computation using multiple CPUs and GPUs including applications on a supercomputer. We illustrate the performance of the surveyed methods using a multiagent model. Many codes are publicly available.},
  issn          = {1574-0021},
  keywords      = {High dimensions, Large scale, Projection, Perturbation, Stochastic simulation, Value function iteration, Endogenous grid, Envelope condition, Smolyak, -distinguishable set, Curse of dimensionality, Precomputation, Manifold, Parallel computation, Supercomputers},
}

@Article{GriebelHarbrechtMulterer2015,
  author         = {M. Griebel and H. Harbrecht and M. Multerer},
  title          = {Multilevel quadrature for elliptic parametric partial differential equations in case of polygonal approximations of curved domains},
  journal        = {SIAM Journal on Numerical Analysis},
  year           = {2020},
  volume         = {58},
  number         = {1},
  pages          = {684--705},
  __markedentry  = {[alexa:1]},
  annote         = {article},
  inspreprintnum = {1521},
}

@Book{rao1991theory,
  author        = {Rao, M. and Ren, Z.},
  title         = {Theory of {O}rlicz {S}paces},
  year          = {1991},
  series        = {Chapman \& Hall Pure and Applied Mathematics},
  publisher     = {Taylor \& Francis},
  __markedentry = {[alexa:1]},
  lccn          = {91007684},
}

@InCollection{hajivassiliou1994classical,
  author        = {Hajivassiliou, V. and Ruud, P.},
  title         = {Classical estimation methods for {LDV} models using simulation},
  booktitle     = {Handbook of {E}conometrics},
  year          = {1994},
  volume        = {4},
  publisher     = {Elsevier},
  chapter       = {40},
  pages         = {2383 - 2441},
  __markedentry = {[alexa:1]},
  abstract      = {Publisher Summary
This chapter discusses classical estimation methods for limited dependent variable (LDV) models that employ Monte Carlo simulation techniques to overcome computational problems in such models. These difficulties take the form of high-dimensional integrals that need to be calculated repeatedly. In the past, investigators were forced to restrict attention to special classes of LDV models that are computationally manageable. The simulation estimation methods we discuss here make it possible to estimate LDV models that are computationally intractable using classical estimation methods. The chapter first reviews the ways in which LDV models arise, describing the differences and similarities in censored and truncated data generating processes. Censoring and truncation give rise to the troublesome multivariate integrals. Following the LDV models, we described various simulation methods for evaluating such integrals. Naturally, censoring and truncation play roles in simulation as well. Finally, estimation methods that rely on simulation are described. The chapter also reviews three general approaches that combine estimation of LDV models and simulation: simulation of the log-likelihood function (MSL), simulation of moment functions (MSM), and simulation of the score (MSS). The MSS is a combination of ideas from MSL and MSM, treating the efficient score of the log-likelihood function as a moment function.},
  issn          = {1573-4412},
}

@InCollection{keane2011structural,
  author        = {Keane, M. and Todd, P. and Wolpin, K.},
  title         = {The Structural Estimation of Behavioral Models: {D}iscrete Choice Dynamic Programming Methods and Applications},
  booktitle     = {Handbook of {L}abor {E}conomics},
  year          = {2011},
  editor        = {Ashenfelter, O. and Card, D.},
  volume        = {4},
  publisher     = {Elsevier},
  chapter       = {4},
  pages         = {331-461},
  __markedentry = {[alexa:1]},
  issn          = {1573-4463},
  keywords      = {Structural estimation, Discrete choice, Dynamic programming, Labor supply, Job search, Human capital},
}

@Article{Bhat1998,
  author        = {Bhat, C.},
  title         = {Accommodating variations in responsiveness to level-of-service measures in travel mode choice modeling},
  journal       = {Transportation Research Part A: Policy and Practice},
  year          = {1998},
  volume        = {32},
  number        = {7},
  pages         = {495 - 507},
  __markedentry = {[alexa:1]},
  keywords      = {intercity travel, maximum simulated likelihood function, multinomial logit model, taste heterogeneity, random-coefficients.},
}

@Article{hajivassiliou1994b,
  author        = {Hajivassiliou, V.},
  title         = {A simulation estimation analysis of the external debt crises of developing countries},
  journal       = {Journal of Applied Econometrics},
  year          = {1994},
  volume        = {9},
  number        = {2},
  pages         = {109-131},
  __markedentry = {[alexa:1]},
  abstract      = {Abstract In this paper we develop models of the incidence and extent of external financing crises of developing countries, which lead to multiperiod multinomial discrete choice and discrete/continuous econometric specifications with flexible correlation structures in the unobservables. We show that estimation of these models based on simulation methods has attractive statistical properties and is computationally tractable. Three such simulation estimation methods are exposited, analysed theoretically, and used in practice: a method of smoothly simulated maximum likelihood (SSML) based on a smooth recursive conditioning simulator (SRC), a method of simulated scores (MSS) based on a Gibbs sampling simulator (GSS), and an MSS estimator based on the SRC simulator. The data set used in this study comprises 93 developing countries observed through the 1970–88 period and contains information on external financing responses that was not available to investigators in the past. Moreover, previous studies of external debt problems had to rely on restrictive correlation structures in the unobservables to overcome otherwise intractable computational difficulties. The findings show that being able for the first time to allow for flexible correlation patterns in the unobservables through estimation by simulation has a substantial impact on the parameter estimates obtained from such models. This suggests that past empirical results in this literature require a substantial re-evaluation.},
}

@Article{BLP1995,
  author        = {Berry, S. and Levinsohn, J. and Pakes, A.},
  title         = {Automobile Prices in Market Equilibrium},
  journal       = {Econometrica},
  year          = {1995},
  volume        = {63},
  number        = {4},
  pages         = {841--890},
  __markedentry = {[alexa:1]},
  abstract      = {This paper develops techniques for empirically analyzing demand and supply in differentiated products markets and then applies these techniques to analyze equilibrium in the U.S. automobile industry. Our primary goal is to present a framework which enables one to obtain estimates of demand and cost parameters for a class of oligopolistic differentiated products markets. These estimates can be obtained using only widely available product-level and aggregate consumer-level data, and they are consistent with a structural model of equilibrium in an oligopolistic industry. When we apply the techniques developed here to the U.S. automobile market, we obtain cost and demand parameters for (essentially) all models marketed over a twenty year period.},
  publisher     = {[Wiley, Econometric Society]},
}

@InCollection{FERNANDEZVILLAVERDE2016,
  author        = {J. Fernandez-Villaverde and J. Rubio-Ramírez and F. Schorfheide},
  title         = {Solution and estimation methods for {DSGE} models},
  booktitle     = {Handbook of {M}acroeconomics},
  year          = {2016},
  editor        = {J. Taylor and H. Uhlig},
  volume        = {2},
  publisher     = {Elsevier},
  chapter       = {9},
  pages         = {527 - 724},
  __markedentry = {[alexa:1]},
  abstract      = {This chapter provides an overview of solution and estimation techniques for dynamic stochastic general equilibrium models. We cover the foundations of numerical approximation techniques as well as statistical inference and survey the latest developments in the field.},
  issn          = {1574-0048},
  keywords      = {Approximation error analysis, Bayesian inference, DSGE model, Frequentist inference, GMM estimation, Impulse response function matching, Likelihood-based inference, Metropolis-Hastings algorithm, Minimum distance estimation, Particle filter, Perturbation methods, Projection methods, Sequential Monte Carlo., C11, C13, C32, C52, C61, C63, E32, E52},
}

@InProceedings{GerstnerHeinz2013,
  author        = {Gerstner, T. and Heinz, S.},
  title         = {Dimension- and time-adaptive multilevel {M}onte {C}arlo methods},
  booktitle     = {Sparse {G}rids and {A}pplications},
  year          = {2013},
  editor        = {J.~Garcke and M.~Griebel},
  volume        = {88},
  series        = {Lecture Notes in Computational Science and Engineering},
  pages         = {107-120},
  __markedentry = {[alexa:1]},
  journal       = {Lecture Notes in Computational Science and Engineering},
}

@InCollection{AdamsFournier2003,
  author        = {R. Adams and J. Fournier},
  title         = {Orlicz {S}paces and {O}rlicz-{S}obolev {S}paces},
  booktitle     = {Sobolev Spaces},
  year          = {2003},
  volume        = {140},
  series        = {Pure and Applied Mathematics},
  publisher     = {Elsevier},
  chapter       = {8},
  pages         = {261 - 294},
  __markedentry = {[alexa:1]},
  issn          = {0079-8169},
}

@MastersThesis{Gilch2020,
  author        = {Gilch, A.},
  title         = {Applications of higher-order quadrature methods to econometric models and estimators},
  year          = {2020},
  type          = {Master Thesis, Institut f\"ur Numerische Simulation, Universit\"at Bonn},
  __markedentry = {[alexa:1]},
  annote        = {INSmaster},
  cosupervisor  = {B. Bohn},
  supervisor    = {M. Griebel},
}

@Article{KaUlVo19,
  author        = {L. K\"ammerer and T. Ullrich and T. Volkmer},
  title         = {Worst case recovery guarantees for least squares approximation using random samples},
  journal       = {arXiv e-prints},
  year          = {2019},
  pages         = {1--47},
  __markedentry = {[alexa:1]},
  annote        = {preprint, group},
  file          = {1911.10111.pdf:https\://arxiv.org/pdf/1911.10111.pdf:PDF},
}

@Article{DungUllrich2015,
  author        = {Dung, Dinh and Ullrich, Tino},
  title         = {Lower bounds for the integration error for multivariate functions with mixed smoothness and optimal {F}ibonacci cubature for functions on the square},
  journal       = {Mathematische Nachrichten},
  year          = {2015},
  volume        = {288},
  number        = {7},
  pages         = {743-762},
  __markedentry = {[alexa:1]},
  abstract      = {We prove lower bounds for the error of optimal cubature formulae for d-variate functions from Besov spaces of mixed smoothness in the case , and , where is either the d-dimensional torus or the d-dimensional unit cube . In addition, we prove upper bounds for QMC integration on the Fibonacci-lattice for bivariate periodic functions from in the case , and . A non-periodic modification of this classical formula yields upper bounds for if . In combination these results yield the correct asymptotic error of optimal cubature formulae for functions from and indicate that a corresponding result is most likely also true in case . This is compared to the correct asymptotic of optimal cubature formulae on Smolyak grids which results in the observation that any cubature formula on Smolyak grids can never achieve the optimal worst-case error.},
  keywords      = {Quasi-Monte-Carlo integration, Besov spaces of mixed smoothness, Fibonacci lattice, B-spline representations, Smolyak grids, 41A55, 65D32, 41A25, 41A58, 41A63},
}

@Book{brasspetras2011,
  author    = {Brass, Helmut and Petras, Knut},
  title     = {Quadrature {T}heory: {T}he {T}heory of {N}umerical {I}ntegration on a {C}ompact {I}nterval},
  year      = {2011},
  publisher = {American Mathematical Soc.},
}

@Book{berlinet,
  author        = {A. Berlinet and C. Thomas-Agnan},
  title         = {Reproducing {K}ernel {H}ilbert {S}paces in {P}robability and {S}tatistics},
  year          = {2004},
  publisher     = {Springer},
  __markedentry = {[alexa:1]},
}

@Book{DuTeUl2015,
  author        = {{D}ung, {D}inh and {T}emlyakov, {V}ladimir  and {U}llrich, {T}ino},
  title         = {{H}yperbolic {C}ross {A}pproximation},
  year          = {2018},
  series        = {Advanced Courses in Mathematics. CRM Barcelona.},
  publisher     = {Birkhäuser/Springer Basel},
  __markedentry = {[alexa:1]},
  journal       = {Advanced Courses in Mathematics. CRM Barcelona. BirkhÃ¤user/Springer},
}
\end{document}